\documentclass{conm-p-l}

\usepackage{geometry}
\usepackage{amssymb}
\usepackage{amsthm}
\usepackage{amsmath}
\usepackage[dvips]{epsfig}
\usepackage{graphicx}
\usepackage{color}
\usepackage{url}
\usepackage{tikz}
\usepackage{caption}
\usepackage{subcaption}
\usepackage{epstopdf}
\usepackage[arrow, matrix, curve]{xy}

\usetikzlibrary{decorations.markings}

\makeatletter\@addtoreset {equation}{section}\makeatother

\usetikzlibrary{matrix,arrows}
\newtheorem{thm}{Theorem}
\newtheorem{cor}{Corollary}
\newtheorem{prop}{Proposition}
\newtheorem{lem}{Lemma}
\newtheorem{mydef}{Definition}
\newtheorem{remark}{Remark}

{\begin{trivlist} \item[]{\em Proof }}%
{\hspace*{\fill}$\rule{.3\baselineskip}{.35\baselineskip}$\end{trivlist}}

\DeclareMathOperator{\C}{\mathbb{C}}

\DeclareMathOperator{\R}{\mathbb{R}}

\numberwithin{equation}{section}

\begin{document}

\title[]{The derivative NLS equation: global existence with solitons}

\author{Dmitry E. Pelinovsky}
\address[D. Pelinovsky]{Department of Mathematics, McMaster University, Hamilton, Ontario, Canada, L8S 4K1}
\email{dmpeli@math.mcmaster.ca}

\author{Aaron Saalmann}
\address[A. Saalmann]{Mathematisches Institut, Universit\"{a}t zu K\"{o}ln, 50931 K\"{o}ln, Germany}
\email{asaalman@math.uni-koeln.de}
\thanks{A.S. gratefully acknowledges financial support from the projects ``Quantum Matter and Materials"
and SFB-TRR 191 ``Symplectic Structures in Geometry, Algebra and Dynamics" (Cologne University, Germany).}

\author{Yusuke Shimabukuro}
\address[Y. Shimabukuro]{Institute of Mathematics, Academia Sinica, Taipei, Taiwan, 10617}
\email{shimaby@gate.sinica.edu.tw}

\thanks{}

\date{}

\dedicatory{}


\begin{abstract}
We prove the global existence result for the derivative NLS equation
  in the case when the initial datum includes a finite number of solitons.
  This is achieved by an application of the B\"{a}cklund transformation that
  removes a finite number of zeros of the scattering coefficient.
  By means of this transformation, the Riemann--Hilbert problem for meromorphic
  functions can be formulated as the one for analytic functions,
  the solvability of which was obtained recently. A major difficulty in the proof is
  to show invertibility of the B\"{a}cklund transformation acting on weighted Sobolev spaces.
\end{abstract}

\maketitle

\section{Introduction}

We consider the Cauchy problem for the derivative nonlinear Schr\"{o}dinger (DNLS) equation
\begin{equation}\label{dnls}
\left\{ \begin{array}{l} iu_t+u_{xx} + i(|u|^2 u)_x = 0, \quad t \in \R,\\
u|_{t=0} = u_0, \end{array} \right.
\end{equation}
where the subscripts denote partial derivatives and $u_0$ is defined in $H^2(\mathbb{R}) \cap H^{1,1}(\mathbb{R})$.
Here $H^m(\mathbb{R})$ denotes the Sobolev space of distributions with square integrable derivatives up to the order $m$,
$H^{1,1}(\mathbb{R})$ denotes the weighted Sobolev space given by
\begin{equation*}
H^{1,1}(\mathbb{R}) = \left\{ u \in L^{2,1}(\mathbb{R}), \quad \partial_x u \in L^{2,1}(\mathbb{R}) \right\},
\end{equation*}
and the weighted space $L^{2,1}(\mathbb{R})$ is equipped with the norm
$$
\| u \|_{L^{2,1}} = \left( \int_{\mathbb{R}} \langle x \rangle^{2} |u|^2 dx \right)^{1/2}, \quad
\langle x \rangle := (1 + x^2)^{1/2}.
$$

Global well-posedness of the Cauchy problem (\ref{dnls}) for $u_0$ in $H^2(\mathbb{R})$
was shown for initial datum with small $H^1(\mathbb{R})$ norm in the pioneer works of
Tsutsumi \& Fukuda \cite{TF1,TF2}. Hayashi \cite{Hayashi} and
Hayashi \& Ozawa \cite{H-O-1} extended the global well-posedness for $u_0$ in $H^1(\R)$
with small $L^2(\mathbb{R})$ norm. The critical $L^2(\mathbb{R})$ norm corresponds
to the stationary solitary waves of the DNLS equation. The question of whether global solutions
for initial datum with large $L^2(\mathbb{R})$ norm exist in the Cauchy problem
(\ref{dnls}) was addressed very recently by using different analytical and numerical methods.

Wu \cite{W,W-2} combined the mass, momentum and energy conservation with variational arguments
and pushed up the upper bound on the $L^2(\mathbb{R})$-norm of the initial datum required
for existence of global solutions. By adding a new result on orbital stability of
algebraically decaying solitons \cite{WuKwon}, this upper bound is pushed up even higher,
but still within the range of the $L^2(\mathbb{R})$ norm of the travelling
solitary waves of the DNLS equation.

Orbital stability of one-soliton solutions was shown long ago by Guo \& Wu \cite{GuoWu1995} and
Colin \& Ohta \cite{CO06}.
More recently, the orbital stability of multi-soliton solutions was obtained in the energy space,
under suitable assumptions on the speeds and frequencies of the single solitons \cite{LeCoz}.
Variational characterization of the DNLS solitary waves and further improvements of the global existence
near a single solitary wave were developed in \cite{MWX1}. Orbital stability of a sum
of two solitary waves was obtained from the variational characterization in \cite{MWX2} (see also
\cite{FHT,MWX3} for similar analysis of the generalized DNLS equation).

Numerical simulations of the DNLS equation (\ref{dnls}) indicate no blow-up phenomenon
for initial data in $H^1(\mathbb{R})$ with any large $L^2(\mathbb{R})$ norm \cite{Sulem1,Sulem2}.
The same conclusion is confirmed by means of the asymptotic analysis of the self-similar blow-up
solutions \cite{Sulem3}.

Since the DNLS equation (\ref{dnls}) is formally solvable with the inverse scattering transform
method \cite{KN78}, one can look at other analytical
tools to deal with the same question.
Lipschitz continuity of the direct and inverse scattering transform in appropriate
function spaces was established very recently \cite{Perry,Pelinovsky-Shimabukuro}
and this result suggests global well-posedness of the Cauchy problem
(\ref{dnls}) without sharp constraints on the $L^2(\mathbb{R})$ norm of the initial datum.
The solvability of the inverse scattering transform was achieved
by using the pioneer results of Deift \& Zhou \cite{D-Z-1} and Zhou \cite{Z}
but extended from the Zakharov--Shabat (ZS) to the Kaup--Newell (KN) spectral problem.
Simplifying assumptions were made in \cite{Perry,Pelinovsky-Shimabukuro}
to exclude eigenvalues and resonances in the KN spectral problem.
Excluding resonances is a natural condition to define so-called generic initial data $u_0$.
On the other hand,
eigenvalues are usually excluded if the initial datum satisfies the small-norm constraint,
and it is not obvious if there exist the initial datum $u_0$
with a large $L^2(\mathbb{R})$ norm that yield no eigenvalues in the KN spectral problem.

The goal of the present paper is to extend the result from \cite{Pelinovsky-Shimabukuro}
to the case of a finite number of eigenvalues in the KN spectral problem. Working
with the B\"{a}cklund transformation, similarly to the work \cite{C-P,D-J} for the ZS spectral
problem, we are able to apply the inverse scattering transform technique to the initial datum
with a finite number of solitons. By using the solvability
result from \cite{Pelinovsky-Shimabukuro} and the invertibility of the B\"{a}cklund transformation proved here,
we are able to extend the global well-posedness result for the Cauchy problem (\ref{dnls})
to arbitrarily large initial data in $H^2(\mathbb{R}) \cap H^{1,1}(\mathbb{R})$.

The main algebraic tool used in this paper is definitely not new. Imai \cite{Imai}
used the multi-fold B\"{a}cklund transformation to obtain multi-solitons and quasi-periodic solutions
of the DNLS equation. Steudel \cite{Steudel} gave a very nice overview of the construction of
the multi-solitons with the same technique. More recent treatments of
the B\"{a}cklund transformations for the DNLS equation can be found in
further works \cite{Guo-Ling-Liu,Xu-He-Wang}. What makes this present paper new
is the way how the B\"{a}cklund transformation can be applied in
the rigorous treatment of the inverse scattering transform
and the global well-posedness problem.

The DNLS equation appears to be a compatibility condition for $C^2$
solutions to the KN spectral problem
\begin{equation} \label{laxeq1}
\partial_x \psi = \left[ -i\lambda^2\sigma_3+\lambda Q(u) \right] \psi
\end{equation}
and the time-evolution problem
\begin{equation}\label{laxeq2}
\partial_t \psi = \left[ - 2i\lambda^4\sigma_3 + 2 \lambda^3 Q(u) + i \lambda^2 |u|^2 \sigma_3
- \lambda |u|^2 Q(u) + i \lambda \sigma_3 Q(u_x) \right] \psi,
\end{equation}
where $\lambda \in \mathbb{C}$ is the $(t,x)$-independent spectral parameter,
$\psi(t,x)$ is the $\mathbb{C}^2$ vector for the wave function, and
$Q(u)$ is the $(t,x)$-dependent matrix potential given by
\begin{equation}
\label{lax-matrices}
Q(u) = \left[\begin{matrix}0&u\\ -\overline{u} & 0\end{matrix}\right].
\end{equation}
The Pauli matrices that include $\sigma_3$ are given by
\begin{equation}
\label{Pauli}
\sigma_1 = \left[ \begin{matrix} 0 & 1 \\ 1 & 0 \end{matrix} \right], \quad
\sigma_2 = \left[\begin{matrix} 0&i\\ -i & 0\end{matrix}\right], \quad
\sigma_3 = \left[\begin{matrix} 1&0\\ 0 & -1\end{matrix}\right].
\end{equation}
A long but standard computation
shows that the compatibility condition $\partial_t\partial_x \psi =\partial_x\partial_t \psi$
for $C^2$ solutions of system (\ref{laxeq1}) and (\ref{laxeq2})
is equivalent to the DNLS equation $i u_t + u_{xx} + i (|u|^2 u)_x = 0$ for classical solutions $u$.

The following theorem presents the main result.

\begin{thm}
For every $u_0 \in H^2(\mathbb{R}) \cap H^{1,1}(\mathbb{R})$ such that the
KN spectral problem (\ref{laxeq1}) admits no resonances in the sense of Definition \ref{definition-resonance}
and only simple eigenvalues in the sense of Definition \ref{definition-eigenvalue},
there exists a unique global solution $u(t,\cdot) \in H^2(\mathbb{R}) \cap H^{1,1}(\mathbb{R})$
of the Cauchy problem \eqref{dnls} for every $t \in \mathbb{R}$. 
\label{theorem-main}
\end{thm}

The organization of the paper is as follows. Section 2 contains the review of
Jost functions and scattering coefficients from \cite{Pelinovsky-Shimabukuro}.
Section 3 presents the B\"{a}cklund transformation for the KN spectral problem
in the form suitable for our analysis. Section 4 adds the time evolution
for the B\"{a}cklund transformation according to the DNLS equation.
Section 5 gives an example of the B\"{a}cklund transformation connecting the one-soliton
and zero-soliton solutions. Section 6 completes the proof of Theorem \ref{theorem-main}.
Appendix A lists useful properties of operators used in the definition
of the B\"{a}cklund transformation. Appendix B gives a technical result
on the regularity of Jost functions for the KN spectral problem.

\section{Review of the direct scattering transform}

We introduce Jost functions for the KN spectral problem (\ref{laxeq1})
under some conditions on the potential $u$. In this section, we freeze the time variable $t$ and drop
it from the argument list of the dependent functions. The following two propositions were proved
in the previous work (see Lemma 1, Corollary 2, and Corollary 3 in \cite{Pelinovsky-Shimabukuro}). Here
$e_{1,2}$ are standard basis vectors in $\mathbb{R}^2$.

\begin{prop}\label{KN-existence}
Let $u \in L^1(\mathbb{R}) \cap L^{\infty}(\mathbb{R})$ and $\partial_x u \in L^1(\mathbb{R})$.
For every $\lambda \in \mathbb{R} \cup i \mathbb{R}$, there exist unique solutions
$\varphi_{\pm}(x;\lambda) e^{-i\lambda^2x}$ and $\phi_{\pm}(x;\lambda) e^{i\lambda^2x}$
to the KN spectral problem \eqref{laxeq1} with
 $\varphi_{\pm}(\cdot;\lambda) \in L^{\infty}(\mathbb{R})$ and
$\phi_{\pm}(\cdot;\lambda) \in L^{\infty}(\mathbb{R})$ such that
\begin{equation}\label{jost-infinity-original}
\left. \begin{array}{l}
\varphi_{\pm}(x;\lambda) \to e_1, \\
\phi_{\pm}(x;\lambda) \to e_2, \end{array} \right\} \quad \mbox{\rm as} \quad x \to \pm \infty.
\end{equation}
\end{prop}

\begin{prop}
\label{prop-Jost-KN}
Under the same assumption on $u$ as in Proposition \ref{KN-existence}, for every $x \in \mathbb{R}$, the Jost functions
$\varphi_{-}(x;\cdot)$ and $\phi_{+}(x;\cdot)$ are analytic in the first and third quadrant of the $\lambda$ plane
(where ${\rm Im}(\lambda^2) > 0$), whereas the functions
$\varphi_{+}(x;\cdot)$ and $\phi_{-}(x;\cdot)$ are analytic in the second and fourth quadrant of the $\lambda$ plane
(where ${\rm Im}(\lambda^2) < 0$). Furthermore, for every $\lambda$ with ${\rm Im}(\lambda^2) > 0$ and for all
$u$ satisfying $\|u\|_{L^1 \cap L^{\infty}} + \| \partial_x u \|_{L^1} \leq M$ there exists a constant $C_M$ which does not depend on $u$, such that
\begin{equation}\label{KN L^infty estimate}
  \|\varphi_{-}(\cdot;\lambda) \|_{L^{\infty}} + \|\phi_{+}(\cdot;\lambda) \|_{L^{\infty}}\leq C_M.
\end{equation}
\end{prop}

The set of Jost functions $[\varphi_-(x;\lambda),\psi_-(x;\lambda)] e^{-i\lambda^2\sigma_3x}$
at the left infinity is linearly dependent from the set of Jost functions
$[\varphi_+(x;\lambda),\psi_+(x;\lambda)] e^{-i\lambda^2\sigma_3x}$
at the right infinity. Therefore,
for every $\lambda \in \mathbb{R} \cup i \mathbb{R}$ there exists
the transfer matrix $S(\lambda)$ that connects the two sets as follows:
\begin{equation} \label{linear-comb}
[\varphi_-(x;\lambda),\phi_-(x;\lambda)] e^{-i\lambda^2\sigma_3x} =
[\varphi_+(x;\lambda),\phi_+(x;\lambda)] e^{-i\lambda^2\sigma_3x} S(\lambda),
\end{equation}
where $x \in \mathbb{R}$ is arbitrary. Thanks to the symmetry relations
\begin{equation}
\label{symmetry-relation}
\phi_{\pm}(x;\lambda)=\sigma_1\sigma_3\overline{\varphi_{\pm}(x;\overline{\lambda})},
\end{equation}
the transfer matrix $S$ has the structure
\begin{equation}
\label{transfer-matrix}
S(\lambda) = \left[\begin{matrix} a(\lambda) &  -\overline{b(\overline{\lambda})}\\b(\lambda) & \overline{a(\overline{\lambda})} \end{matrix}\right],
\end{equation}
defined by the two scattering coefficients $a(\lambda)$ and $b(\lambda)$. Since the determinant
of the transfer matrix $S(\lambda)$ is equal to unity for every $\lambda \in \mathbb{R} \cup i \mathbb{R}$,
we have the following relation between $a(\lambda)$ and $b(\lambda)$:
\begin{eqnarray}
\label{W-a-b-4}
a(\lambda)\overline{a(\overline{\lambda})} + b(\lambda)\overline{b(\overline{\lambda})} = 1, \quad \lambda \in \mathbb{R} \cup i \mathbb{R}.
\end{eqnarray}
Furthermore, scattering coefficients $a(\lambda)$ and $b(\lambda)$ can be written in terms of Jost functions
by using the Wronskian determinant $W(\eta,\xi)=\eta_1\xi_2-\xi_1\eta_2$ defined for $\eta,\xi \in \mathbb{C}^2$:
\begin{subequations}\label{a-b}
  \begin{eqnarray}
    a(\lambda) &=& W(\varphi_-(x;\lambda) e^{-i \lambda^2 x},\phi_+(x;\lambda) e^{+i \lambda^2 x}), \label{a} \\
    b(\lambda) &=& W(\varphi_+(x;\lambda) e^{-i \lambda^2 x},\varphi_-(x;\lambda) e^{-i \lambda^2 x}). \label{b}
  \end{eqnarray}
\end{subequations}
The coefficient $b(\lambda)$ is expressed by the Jost functions whose analytic domains in the $\lambda$ plane are disjoint.
As a result, $b(\lambda)$ cannot be continued into the complex plane of $\lambda$. On the other hand, $a(\lambda)$
can be continued analytically into the complex plane of $\lambda$,
according to the following result (see Lemma 4 in \cite{Pelinovsky-Shimabukuro}).

\begin{prop}
\label{prop-scattering}
Under the same assumption on $u$ as in Proposition \ref{KN-existence},
the scattering coefficient $a(\lambda)$ can be continued analytically into $\{\lambda \in \C: \mbox{Im}(\lambda^2)>0\}$ with the limit
$$
a_{\infty} = \lim_{|\lambda|\rightarrow \infty}a(\lambda)=e^{\frac{1}{2i}\|u\|^2_{L^2}}.
$$
Similarly, $\overline{a(\overline{\lambda})}$ is continued
analytically into $\{\lambda \in \C: \mbox{Im}(\lambda^2)<0\}$ with the limit
$$
\overline{a}_{\infty} = \lim_{|\lambda|\rightarrow \infty}\overline{a(\overline{\lambda})}=e^{-\frac{1}{2i}\|u\|^2_{L^2}}.
$$
\end{prop}

Since $a_{\infty}\neq 0$ if $u \in L^2(\mathbb{R})$, the following
corollary holds by a theorem of complex analysis on zeros of analytic functions.

\begin{cor}\label{cor finite many zeros}
  Under the same assumption on $u$ as in Proposition \ref{KN-existence},
  the scattering coefficient $a(\lambda)$ has at most finite number of zeros in $\{\lambda \in \C: \mbox{Im}(\lambda^2)>0\}$.
\end{cor}

If a potential $u$ is sufficiently small, then one can easily deduce that $a(\lambda)$ has no zeros
in the domain of its analyticity. As is explained in Remark 5 in \cite{Pelinovsky-Shimabukuro},
$a(\lambda) \neq 0$ for every $\mbox{Im}(\lambda^2) \geq 0$ if
$$
\|u\|_{L^2}^2+\sqrt{\|u\|_{L^1} (2\|\partial_xu\|_{L^1}+\|u\|_{L^3}^3)}<1.
$$
However, for sufficiently large $u$, the spectral coefficient $a(\lambda)$ may have zeros
for some $\mbox{Im}(\lambda^2) \geq 0$. We distinguish two cases, according to the following definitions.

\begin{mydef}
\label{definition-resonance}
If $a(\lambda_0) = 0$ for $\lambda_0 \in \mathbb{R} \cup i \mathbb{R}$,
we say that $\lambda_0$ is a resonance of the spectral problem (\ref{laxeq1}).
\end{mydef}

\begin{mydef}
\label{definition-eigenvalue}
If $a(\lambda_0) = 0$ for $\lambda_0 \in \mathbb{C}_I := \{ {\rm Re}(\lambda) > 0, \quad {\rm Im}(\lambda) > 0 \}$,
we say that $\lambda_0$ is an eigenvalue of the spectral problem (\ref{laxeq1}) in $\mathbb{C}_I$.
An eigenvalue is called simple if $a'(\lambda_0) \neq 0$.
\end{mydef}

\begin{remark}
By the symmetry of the KN spectral problem (\ref{laxeq1}),
if $a(\lambda_0) = 0$ for $\lambda_0 \in \mathbb{C}_I$, then $a(-\lambda_0) = 0$.
\end{remark}

\begin{remark}
If $u \in H^{1,1}(\mathbb{R})$, then the assumption of Propositions \ref{KN-existence}, \ref{prop-Jost-KN}, and \ref{prop-scattering}
are satisfied so that $u \in L^1(\mathbb{R}) \cap L^{\infty}(\mathbb{R})$ and $\partial_x u \in L^1(\mathbb{R})$.
To enable the inverse scattering transform, we will work with $u$ in $H^2(\mathbb{R}) \cap H^{1,1}(\mathbb{R})$.
\end{remark}

The main assumption of Theorem \ref{theorem-main} excludes resonances but includes simple eigenvalues.
Thanks to Corollary \ref{cor finite many zeros}, the number of eigenvalues is finite
under the assumptions in Proposition \ref{KN-existence}.
Therefore, the initial datum $u_0$ of the Cauchy problem (\ref{dnls}) in $H^2(\mathbb{R}) \cap H^{1,1}(\mathbb{R})$
may include at most finitely many solitons.

Let $Z_N$ be a subset of $H^2(\mathbb{R}) \cap H^{1,1}(\mathbb{R})$ such that $a(\lambda)$ has $N$ simple
zeros in the first quadrant $\C_I$. Zeros of $a(\lambda)$ are assumed to be simple
in order to simplify our presentation. This is not a restricted assumption because
the union of $\{Z_N\}_{N \in \mathbb{N}}$ is dense in space $H^2(\mathbb{R}) \cap H^{1,1}(\mathbb{R})$ thanks to
the classical result of Beals \& Coifman \cite{Beals1984}.
Indeed, as is known from \cite{KN78} (see also \cite{Pelinovsky-Shimabukuro}),
the Kaup-Newell spectral system (\ref{laxeq1}) can be transformed to the
Zakharov--Shabat spectral system by the transformation
$$
\widetilde{\psi}(x) =
\left[
  \begin{array}{cc}
    e^{\frac{1}{2i} \int_{x}^{\infty}|u(y)|^2dy} & 0 \\
    0 & e^{-\frac{1}{2i} \int_{x}^{\infty}|u(y)|^2dy} \\
  \end{array}
\right]
\left[
  \begin{array}{cc}
    1 & 0 \\
    -\overline{u}(x) & 2i\lambda \\
  \end{array}
\right] \psi(x),
$$
where $\widetilde{\psi}$ satisfies
\begin{equation}\label{ZS}
        \partial_x \widetilde{\psi} = \left[ -i\lambda^2\sigma_3+\widetilde{Q}(u) \right] \widetilde{\psi},
\end{equation}
with
\begin{equation*}
        \widetilde{Q}(u) = \frac{1}{2i}
\left[
  \begin{array}{cc}
    0 & u e^{-i \int_{x}^{\infty} |u(y)|^2 dy} \\
    -(2i\overline{u}_x+\overline{u}|u|^2) e^{i \int_{x}^{\infty} |u(y)|^2 dy} & 0 \\
  \end{array}
\right].
\end{equation*}
Eigenvalues of the spectral problems (\ref{laxeq1}) and (\ref{ZS}) coincide
and the potential $\widetilde{Q}(u)$ is now defined in $L^1(\mathbb{R})$
under the assumption that $u \in H^{1,1}(\mathbb{R})$.
Proposition 2.30 in \cite{Beals1984} yields the following result.

\begin{prop}
  The subset $Z:=\bigcup_{N=1}^{\infty}Z_N$ is dense in $H^2(\mathbb{R}) \cap H^{1,1}(\mathbb{R})$.
\end{prop}

Let $u \in Z_N$ and $a(\lambda)$ vanishes at some $\lambda_j \in \C_I$, $j \in \{1,2,...,N\}$.
It follows from the definition of $a(\lambda)$ in \eqref{a} that the Jost functions
$\varphi_-(x;\lambda_j) e^{-i\lambda_j^2 x}$
and $\phi_+(x;\lambda_j) e^{i\lambda_j^2 x}$ are linearly dependent.
This implies that there is a norming coefficient $\gamma_j \in \mathbb{C}$
such that
\begin{equation}
\label{norming-coeff}
\varphi_-(x; \lambda_j) e^{-i\lambda_j^2x} = \gamma_j \phi_+(x;\lambda_j) e^{i\lambda_j^2x}, \quad  x \in \mathbb{R}.
\end{equation}
Since $\varphi_-(x;\lambda_j) e^{-i\lambda_j^2 x}$
and $\phi_+(x;\lambda_j) e^{i\lambda_j^2 x}$  are uniquely determined by Proposition \ref{KN-existence},
the norming coefficient $\gamma_j$ is determined uniquely.

\begin{remark}
Because $\lambda_j \in \mathbb{C}_I$, the Jost functions $\varphi_-(x;\lambda_j) e^{-i\lambda_j^2 x}$
and \newline$\phi_+(x;\lambda_j) e^{i\lambda_j^2 x}$ in (\ref{norming-coeff})
decay to zero as $|x| \to \infty$ exponentially fast. Hence, they represent an eigenvector
of the spectral problem (\ref{laxeq1}) for the simple eigenvalue $\lambda_j$.
\end{remark}

\section{B\"{a}cklund transformation}

\label{section-BT}

In order to define the B\"{a}cklund transformation in the simplest form, let us
introduce the bilinear form $d_{\lambda}$ that acts on $\C^2$ for a fixed $\lambda \in \C$.
If $\eta=(\eta_1,\eta_2)^t$ and $\xi=(\xi_1,\xi_1)^t$ are in $\C^2$, then
\begin{equation}
\label{def-d}
d_{\lambda}(\eta,\xi):=\lambda \eta_1 \overline{\xi}_1+\overline{\lambda}\eta_2\overline{\xi}_2.
\end{equation}
We further introduce
\begin{equation}
\label{def-G-S}
G_{\lambda}(\eta):=\frac{d_{\overline{\lambda}}(\eta,\eta)}{d_{\lambda}(\eta,\eta)} \quad \mbox{\rm and} \quad
S_{\lambda}(\eta):=2i(\lambda^2-\overline{\lambda}^2)\frac{\eta_1 \overline{\eta}_2}{d_{\lambda}(\eta,\eta)}.
\end{equation}
Useful algebraic properties of $d_{\lambda}$, $G_{\lambda}$, and $S_{\lambda}$ are
reviewed in Appendix \ref{algebra}.

The B\"{a}cklund transformation can be expressed by using operators $G_{\lambda}$ and $S_{\lambda}$.
Let us first give an informal definition of the B\"{a}cklund transformation and then make it precise.

Suppose that $u$ is a smooth solution of the DNLS equation and $\eta$ is a smooth nonzero solution
of the KN spectral problem (\ref{laxeq1}) associated with the potential $u$ for a fixed $\lambda \in \C \setminus \{0\}$.
The B\"{a}cklund transformation $\mathbf{B}_{\lambda}(\eta)$ is given as
\begin{equation} \label{backlund}
\mathbf{B}_{\lambda}(\eta) u:= G_{\lambda}(\eta)[-G_{\lambda}(\eta)u+S_{\lambda}(\eta)].
\end{equation}
We intend to show that $\mathbf{B}_{\lambda}(\eta)u$ is another smooth solution of the DNLS equation.
Note that
$$
G_{\lambda}(\eta) = -1 \quad \mbox{\rm and} \quad S_{\lambda}(\eta) = 0 \quad \mbox{\rm if} \quad
\lambda \in \mathbb{R} \cup i \mathbb{R},
$$
which implies $\mathbf{B}_{\lambda}(\eta)u = -u$ in this case.
Therefore, it makes sense to use the B\"{a}cklund transformation (\ref{backlund})
for a value of $\lambda$ outside the continuous spectrum,
e.g., for $\lambda \in \mathbb{C}_I$.

The transformation \eqref{backlund} has been derived by a constructive algorithm in \cite{Xu-He-Wang},
where it is called the 2-fold Darboux transformation. It must be noted that, since $\eta$ depends
on $u$ via the KN spectral problem \eqref{laxeq1}, the transformation (\ref{backlund})
is nonlinear in $u$. The function $\mathbf{B}_{\lambda}(\eta)u$ depends on variables $t$ and $x$,
whereas the value of $\lambda$ is fixed. The quantities $u$, $\eta$, as well as $\lambda \in \C \setminus \{0\}$
affect $\mathbf{B}_{\lambda}(\eta) u$, e.g., depending on $\eta$ and $\lambda$, the transformation
can be used to obtain different families of solutions from the same solution $u$.

Let $u(t,\cdot) \in H^2(\mathbb{R}) \cap H^{1,1}(\mathbb{R})$
be a local solution of the Cauchy problem (\ref{dnls}) defined for $t\in (-T,T)$ for some $T>0$.
Such solutions always exist by the local well-posedness theory \cite{H-O-1}.
Assume that $u(t,\cdot) \in Z_1$ which means that the solution to the DNLS equation contains a single soliton related
to a simple eigenvalue $\lambda_1 \in \mathbb{C}_I$ of the KN spectral problem (\ref{laxeq1}).
By using the B\"{a}cklund transformation (\ref{backlund}) with $\lambda=\lambda_1$ and $\eta$
being an eigenvector of the KN spectral problem (\ref{laxeq1}) for the same $\lambda_1$,
we define $u^{(1)}=\mathbf{B}_{\lambda_1}(\eta)u$ as a function of $(t,x)$.
We would like to show that
\begin{itemize}
\item[(i)] $u^{(1)} \in H^2(\mathbb{R}) \cap H^{1,1}(\mathbb{R})$;
\item[(ii)] $u^{(1)} \in Z_0$, that is, the new solution does not contain solitons;
\item[(iii)] $\mathbf{B}_{\lambda_1}(\eta)$  has the (left) inverse so that $u = \left[ \mathbf{B}_{\lambda_1}(\eta) \right]^{-1} u^{(1)}$;
\item[(iv)] $u^{(1)}(t,\cdot) = \mathbf{B}_{\lambda_1}(\eta(t,\cdot))u(t,\cdot)$ is a solution of the DNLS equation for $t \in (-T,T)$.
\end{itemize}
Properties (i) and (iii) are shown in Lemma \ref{BT-lemma1}.  Property (ii) is shown in Lemma \ref{nonzero-a}.
Property (iv) is shown in Lemma \ref{same_solutions}.

In order to obtain the global well-posedness of the Cauchy problem (\ref{dnls}),
we want to extend an existence time $T$ of the solution $u(t,\cdot) \in Z_1$ to an arbitrary large number.
Importantly, the global existence of the solution $u^{(1)}(t,\cdot) \in Z_0$ is known from the previous
works \cite{Perry,Pelinovsky-Shimabukuro}.

Let $\mathbf{B}_{\lambda_1}(\eta^{(1)})$ be the inverse of $\mathbf{B}_{\lambda_1}(\eta)$ for some function $\eta^{(1)}$,
that is, $$\mathbf{B}_{\lambda_1}(\eta^{(1)})u^{(1)} = u.$$ Although the choice of $\eta^{(1)}$ is generally not unique,
we will show in Lemmas \ref{pull-back-eta} and \ref{eta-by-new-jost}
that $\eta^{(1)}$ can be fixed as a unique linear combination of Jost functions of the KN spectral problem (\ref{laxeq1})
associated with the potential $u^{(1)}$. By analyzing the B\"{a}cklund transformation (\ref{backlund}),
we obtain from Lemma \ref{inverse BT-lemma1} the global estimate in the form
\begin{equation}
\label{estimate-key}
\| u(t,\cdot) \|_{H^2 \cap H^{1,1}} = \| \mathbf{B}_{\lambda_1}(\eta^{(1)}(t,\cdot)) u^{(1)}(t,\cdot) \|_{H^2 \cap H^{1,1}}
\leq C_M,
\end{equation}
for every $u^{(1)}(t,\cdot)$ satisfying $\| u^{(1)}(t,\cdot) \|_{H^2 \cap H^{1,1}} \leq M$,
where the constant $C_M$ depends on $M$ but does not depend on $u^{(1)}$.
Since $\|u^{(1)}(t,\cdot)\|_{H^2 \cap H^{1,1}}$ is finite for all times $t \in \mathbb{R}$
(but may grow as $|t| \to \infty$) by the previous results
\cite{Perry,Pelinovsky-Shimabukuro}, the bound (\ref{estimate-key}) yields the
proof of Theorem \ref{theorem-main} in the case of one soliton. By using recursively the
B\"{a}cklund transformation (\ref{backlund}), the result for any number of solitons follows from the result for
one soliton. Thus, the proof of Theorem \ref{theorem-main} relies on the proof of the properties (i)--(iv),
the unique construction of $\eta^{(1)}$ for the inverse B\"{a}cklund transformation
$\mathbf{B}_{\lambda_1}(\eta^{(1)}) = \left[ \mathbf{B}_{\lambda_1}(\eta) \right]^{-1}$, and the estimate (\ref{estimate-key}).

\subsection{Transformation of potentials}
The following lemma shows that the transformation \eqref{backlund} can be defined
as an invertible operator from $u$ to $u^{(1)}$ in the same function space $H^2(\mathbb{R}) \cap H^{1,1}(\mathbb{R})$.
Since we only use the KN spectral problem (\ref{laxeq1}) here, we drop the time variable $t$
from all function arguments.

\begin{lem} \label{BT-lemma1}
Fix $\lambda_1 \in \mathbb{C}_I$. Given a potential $u \in H^2(\mathbb{R}) \cap H^{1,1}(\mathbb{R})$,
define $\eta(x) := \varphi_-(x;\lambda_1) e^{-i\lambda_1^2 x}$,
where $\varphi_-$ is the Jost function for the KN spectral problem \eqref{laxeq1}
in Propositions \ref{KN-existence} and \ref{prop-Jost-KN}.
Then, $u^{(1)} = \mathbf{B}_{\lambda_1}(\eta)u$ belongs to $H^2(\mathbb{R}) \cap H^{1,1}(\mathbb{R})$.
Moreover, the left inverse of $\mathbf{B}_{\lambda_1}(\eta)$ exists.
\end{lem}

\begin{proof}
First, we notice that $d_{\lambda_1}(\eta,\eta) = 0$ if and
only if $\eta = 0$ because ${\rm Re}(\lambda_1) > 0$. However,
if $\eta(x_0) = 0$ at a point $x_0 \in \mathbb{R}$,
then the system \eqref{laxeq1} suggests $\eta'(x_0) = 0$, which implies $\eta(x) = 0$ for every $x \in \mathbb{R}$.
Since $\varphi_-(x;\lambda_1)$ satisfies the nonzero asymptotic limit (\ref{jost-infinity-original}) as $x \to -\infty$,
then $\eta(x) = \varphi_-(x;\lambda_1) e^{-i \lambda_1^2 x} \neq 0$ and
$d_{\lambda_1}(\eta,\eta) \neq 0$ for every finite $x \in \mathbb{R}$.

In order to deal with the limits as $x \to \pm \infty$,
we note that $G_{\lambda}(\eta) = G_{\lambda}(\varphi_-)$
and $S_{\lambda}(\eta) = S_{\lambda}(\varphi_-)$ by properties
(\ref{Gproperty1}) and (\ref{Sproperty1}) of Appendix \ref{algebra}.
Therefore, it is sufficient to consider $d_{\lambda_1}(\varphi_-,\varphi_-)$
instead of $d_{\lambda_1}(\eta,\eta)$.
If $a(\lambda_1)\neq 0$, we claim that there exists $\varepsilon_0>0$ such that
\begin{equation}\label{d(phi_-,phi_-) estimate}
|d_{\lambda_1}(\varphi_-,\varphi_-)|\geq \varepsilon_0, \quad \mbox{\rm for all \;} x \in \R.
\end{equation}
Indeed, since $d_{\lambda_1}(\varphi_-,\varphi_-)\to \lambda_1$ as $x\to-\infty$
and thanks to the arguments above, $d_{\lambda_1}(\varphi_-,\varphi_-)$ may only vanish
in the limit $x \to +\infty$. However, it follows from the representation (\ref{a}) that
the limit $\phi_+(x;\lambda_1) \to e_2$ as $x\to +\infty$, and the fact that
$\varphi_-(\cdot;\lambda_1) \in L^{\infty}(\R)$ imply that
$\varphi_{-,1}(x;\lambda_1) \to a(\lambda_1)$ as $x\to +\infty$
so that $d_{\lambda_1}(\varphi_-,\varphi_-)\nrightarrow 0$ as $x\to +\infty$.
Therefore, the claim (\ref{d(phi_-,phi_-) estimate}) follows.

By using the triangle inequality, the bounds (\ref{norm1})--(\ref{norm2}) of Appendix \ref{appendix-Jost},
the bound (\ref{d(phi_-,phi_-) estimate}), and $|G_{\lambda_1}(\eta)| = 1$, we obtain
\begin{eqnarray}
\nonumber
  \|u^{(1)}\|_{L^{2,1}} &\leq& \|u\|_{L^{2,1}}+ \|S_{\lambda_1}(\varphi_-)\|_{L^{2,1}} \\
   &\leq& \|u\|_{L^{2,1}}+ 2 \varepsilon_0^{-1} |\lambda_1^2-\overline{\lambda}_1^2|
   \left\|\varphi_{-,1}(\cdot,\lambda_1) \overline{\varphi_{-,2}(\cdot,\lambda_1)}\right\|_{L^{2,1}} < \infty.
   \label{estimate-key-backlund}
\end{eqnarray}
The norms $\|\partial_x  u^{(1)}\|_{L^{2,1}}$ as well as  $\|\partial^2_x u^{(1)}\|_{L^{2}}$ are
estimated similarly with the bounds (\ref{norm1})--(\ref{norm2}) of Appendix \ref{appendix-Jost} and the bound (\ref{d(phi_-,phi_-) estimate}).

If $a(\lambda_1)=0$, the uniform bound (\ref{d(phi_-,phi_-) estimate}) is no longer valid
because $d_{\lambda_1}(\varphi_-,\varphi_-)\to 0$ as $x\to+\infty$.
The estimate (\ref{estimate-key-backlund}) can only be proved on the interval $(-\infty,R)$ with arbitrary $R>0$.
In order to extend the estimate (\ref{estimate-key-backlund}) on the interval $(R,\infty)$,
we use (\ref{norming-coeff}) and write $\eta(x) = \varphi_-(x;\lambda_1) e^{-i \lambda_1^2 x} = \gamma_1
\phi_+(x;\lambda_1) e^{i \lambda_1^2 x}$, so that $u^{(1)} = \mathbf{B}_{\lambda_1}(\varphi_-) u$ can be rewritten
as $u^{(1)} = \mathbf{B}_{\lambda_1}(\phi_+) u$. Since
$d_{\lambda_1}(\phi_+,\phi_+)\to \overline{\lambda}_1$ as $x\to+\infty$,
we repeat the same estimates on the interval $(R,\infty)$ by using the equivalent representation of $u^{(1)}$.

Next, we show the existence of the left inverse for $\mathbf{B}_{\lambda_1}(\eta)u$. Let
$\eta^*$ be a vector function and define
\begin{align*}
u^{(2)} & 
=\mathbf{B}_{\lambda_1}(\eta^*)\mathbf{B}_{\lambda_1}(\eta)u \\
&= -G_{\lambda_1}(\eta^*)^2[-G_{\lambda_1}(\eta)^2u+G_{\lambda_1}(\eta)S_{\lambda_1}(\eta)]+G_{\lambda_1}(\eta^*)S_{\lambda_1}(\eta^*)\\
&=G_{\lambda_1}(\eta^*)^2G_{\lambda_1}(\eta)^2u+G_{\lambda_1}(\eta^*)[-G_{\lambda_1}(\eta^*)G_{\lambda_1}(\eta)S_{\lambda_1}(\eta)+S_{\lambda_1}(\eta^*)].
\end{align*}
$\mathbf{B}_{\lambda_1}(\eta^*)$ is the left inverse of $\mathbf{B}_{\lambda_1}(\eta)u$
if $\eta^*$ satisfies
\begin{equation} \label{eta-cond1}
G_{\lambda_1}(\eta^*)^2G_{\lambda_1}(\eta)^2=1
\end{equation}
and
\begin{equation}\label{eta-cond2}
-G_{\lambda_1}(\eta^*)G_{\lambda_1}(\eta)S_{\lambda_1}(\eta)+S_{\lambda_1}(\eta^*)=0.
\end{equation}
System (\ref{eta-cond1}) and (\ref{eta-cond2}) is satisfied either for
\begin{equation} \label{eta-cond1a}
\overline{G_{\lambda_1}(\eta^*)} = G_{\lambda_1}(\eta), \quad S_{\lambda_1}(\eta^*) = S_{\lambda_1}(\eta)
\end{equation}
or for
\begin{equation} \label{eta-cond2a}
\overline{G_{\lambda_1}(\eta^*)} = -G_{\lambda_1}(\eta), \quad S_{\lambda_1}(\eta^*) = -S_{\lambda_1}(\eta).
\end{equation}
Let us show that the choice (\ref{eta-cond2a}) is impossible if $\lambda_1 \in \mathbb{C}_I$.

Since $\eta(x) = \varphi_-(x;\lambda_1) e^{-i\lambda_1^2 x}$, we have
$G_{\lambda_1}(\eta) \to \overline{\lambda}_1/\lambda_1$ as $x \to -\infty$.
Writing
$$
\overline{G_{\lambda_1}(\eta^*)} = \frac{\lambda_1 \frac{|\eta_1^*|^2}{|\eta_2^*|^2} + \overline{\lambda}_1}{\overline{\lambda}_1 \frac{|\eta_1^*|^2}{|\eta_2^*|^2} +\lambda_1},
$$
we realize that $|\eta_1^*|/|\eta_2^*| \nrightarrow 0$ as $x \to -\infty$, as
it would contradict to the first equation in (\ref{eta-cond2a}) with $\lambda_1 \neq 0$.
From the second equation in (\ref{eta-cond2a}), we can see that
$S_{\lambda_1}(\eta^*) \to 0$ as $x \to -\infty$ because $S_{\lambda_1}(\eta) \to 0$ as
$x \to -\infty$. Since $|\eta_1^*|/|\eta_2^*| \nrightarrow 0$ as $x \to -\infty$,
the limit $S_{\lambda_1}(\eta^*) \to 0$ as $x \to -\infty$ implies that
$|\eta_2^*|/|\eta_1^*| \to 0$ as $x \to -\infty$. This implies that
$\overline{G_{\lambda_1}(\eta^*)} \to \lambda_1/\overline{\lambda}_1$ as $x \to -\infty$,
or in view of the first equation in (\ref{eta-cond2a}), we obtain ${\rm Re}(\lambda_1^2) = 0$.
Since $\lambda_1 \in \mathbb{C}_I$, then ${\rm arg}(\lambda_1) = \pi/4$.
Finally writing $\lambda_1 = |\lambda_1| e^{i \pi/4}$ and using the first
equation in (\ref{eta-cond2a}) yields
$$
|\eta_1^*|^2 |\eta_2|^2 + |\eta_2^*|^2 |\eta_1|^2 = 0,
$$
which cannot be satisfied with $\eta^* \neq 0$. This contradiction eliminates
possibility of the choice (\ref{eta-cond2a}).

Thus, we only have the choice (\ref{eta-cond1a}) to define $\eta^*$ and to satisfy system (\ref{eta-cond1}) and (\ref{eta-cond2}).
Since $\lambda_1 \in \C_I$, the condition $\overline{G_{\lambda_1}(\eta^*)} = G_{\lambda_1}(\eta)$
is equivalently written as
$$
|\eta_1|^2 |\eta_1^*|^2 = |\eta_2|^2 |\eta_2^*|^2,
$$
so that there exists a positive number $k$ such that
\begin{equation}
\label{tech-eq-1}
|\eta_1^*| = k |\eta_2|, \quad |\eta_2^*| = k |\eta_1|.
\end{equation}
On the other hand, the condition $S_{\lambda_1}(\eta^*) = S_{\lambda_1}(\eta)$ yields
$$
\frac{\eta_1 \overline{\eta}_2}{\eta_1^* \overline{\eta}_2^*} =
\frac{\lambda_1 |\eta_1|^2 + \overline{\lambda}_1 |\eta_2|^2}{\lambda_1 |\eta_1^*|^2 + \overline{\lambda}_1 |\eta_2^*|^2},
$$
which transforms after substitution of (\ref{tech-eq-1}) to
\begin{equation}
\label{tech-eq-2}
k^2 \frac{\eta_1 \overline{\eta}_2}{\eta_1^* \overline{\eta}_2^*} =
\frac{\lambda_1 |\eta_1|^2 + \overline{\lambda}_1 |\eta_2|^2}{\lambda_1 |\eta_2|^2 + \overline{\lambda}_1 |\eta_1|^2},
\end{equation}
where the right-hand side is of modulus one.
Combining (\ref{tech-eq-1}) and (\ref{tech-eq-2}), we obtain
the most general solution of the system (\ref{eta-cond1a}) in the form
\begin{equation}\label{eta-cond*}
\eta_1^* = k_1 \overline{\eta}_2, \quad \eta_2^* = k_2 \overline{\eta}_1,
\end{equation}
where $k_1,k_2 \in \C$ satisfying $|k_1| = |k_2|$.
Thus, $\mathbf{B}_{\lambda_1}(\eta^*)$ with $\eta^*$ given by (\ref{eta-cond*})
is the left inverse of the transformation $\mathbf{B}_{\lambda_1}(\eta)$.
\end{proof}

The following lemma specifies a unique choice for the function $\eta^*$ constructed
in the proof of Lemma \ref{BT-lemma1}
and shows that $\eta^*$ is a solution of the KN spectral problem (\ref{laxeq1})
associated with the new potential $u^{(1)} = \mathbf{B}_{\lambda_1}(\eta)u$
and the same value $\lambda_1$.

\begin{lem} \label{pull-back-eta}
Under the same assumption as in Lemma \ref{BT-lemma1}, let $\eta^{(1)}$ be given by
\begin{equation}\label{new-eig1}
\eta_1^{(1)} = \frac{\overline{\eta}_2}{d_{\lambda_1}(\eta,\eta)}, \quad
\eta_2^{(1)} = \frac{\overline{\eta}_1}{d_{\overline{\lambda}_1}(\eta,\eta)}.
\end{equation}
Then, $\eta^{(1)}$ is the solution of the KN spectral problem (\ref{laxeq1})
associated with the potential $u^{(1)} = \mathbf{B}_{\lambda_1}(\eta)u$ and the same value $\lambda_1$.
\end{lem}

\begin{proof}
We recall that $\eta$ is a solution of
\begin{equation}
\label{eta-eq}
\partial_x\eta=[-i\lambda_1^2 \sigma_3+\lambda_1 Q(u)]\eta,
\end{equation}
as follows from the KN spectral problem (\ref{laxeq1}) for $\lambda = \lambda_1$.
By using system (\ref{eta-eq}), we obtain
\begin{equation}\label{d_lambda derivative}
\partial_x d_{\lambda_1}(\eta,\eta) = (\lambda_1^2-\overline{\lambda}_1^2) \left[
u \overline{\eta}_1 \eta_2 - i \lambda_1 |\eta_1|^2 + i \overline{\lambda}_1|\eta_2|^2\right].
\end{equation}
By using (\ref{new-eig1}), (\ref{eta-eq}), and
(\ref{d_lambda derivative}), we obtain
\begin{eqnarray*}
\partial_x \eta_1^{(1)} + i \lambda_1^2 \eta_1^{(1)} & = &
\frac{1}{d_{\lambda_1}(\eta,\eta)} \left[ - \overline{\lambda}_1 u \overline{\eta}_1 + i (\lambda_1^2 - \overline{\lambda}_1^2) \overline{\eta}_2\right]
\\
&& \qquad \qquad - \frac{(\lambda_1^2-\overline{\lambda}_1^2) \overline{\eta}_2}{[d_{\lambda_1}(\eta,\eta)]^2} \left[
u \overline{\eta}_1 \eta_2 - i \lambda_1 |\eta_1|^2 + i \overline{\lambda}_1|\eta_2|^2\right] \\
& = & \frac{\lambda_1 \overline{\eta}_1}{[d_{\lambda_1}(\eta,\eta)]^2} \left[
- u d_{\overline{\lambda}_1}(\eta,\eta) + 2 i (\lambda_1^2-\overline{\lambda}_1^2) \eta_1 \overline{\eta}_2 \right] \\
& = & \lambda_1 u^{(1)} \eta_2^{(1)}.
\end{eqnarray*}
Similarly, we obtain
\begin{eqnarray*}
\partial_x \eta_2^{(1)} - i \lambda_1^2 \eta_1^{(1)}\!\!\! &\! =\! &
\frac{1}{d_{\overline{\lambda}_1}(\eta,\eta)} \left[ \overline{\lambda}_1 \overline{u} \overline{\eta}_2 - i (\lambda_1^2 - \overline{\lambda}_1^2) \overline{\eta}_1\right]
\\
&& \qquad \qquad + \frac{(\lambda_1^2-\overline{\lambda}_1^2) \overline{\eta}_1}{[d_{\overline{\lambda}_1}(\eta,\eta)]^2} \left[
\overline{u} \eta_1 \overline{\eta}_2 + i \overline{\lambda}_1 |\eta_1|^2 - i \lambda_1|\eta_2|^2\right] \\
& = & -\frac{\lambda_1 \overline{\eta}_2}{[d_{\overline{\lambda}_1}(\eta,\eta)]^2} \left[
-\overline{u} d_{\lambda_1}(\eta,\eta) + 2 i (\lambda_1^2-\overline{\lambda}_1^2) \overline{\eta_1} \eta_2 \right] \\
& = & -\lambda_1 \overline{u}^{(1)} \eta_1^{(1)}.
\end{eqnarray*}
Thus, we have proven that $\eta^{(1)}$ satisfies the KN spectral problem (\ref{laxeq1}) with the potential
$u^{(1)}$ and the same value $\lambda = \lambda_1$.
\end{proof}

In the construction of Lemmas \ref{BT-lemma1} and \ref{pull-back-eta},
the Jost function $\varphi_-$ was used in the choice for $\eta$.
The following lemma shows that the same potential $u^{(1)}$ can be equivalently
obtained from all four Jost functions of the KN spectral problem (\ref{laxeq1})
if $\lambda_1$ is selected to be a root of the scattering coefficient $a(\lambda)$.

\begin{lem} \label{BT-unique}
Assume that $\lambda_1 \in \mathbb{C}_I$ is chosen so that $a(\lambda_1)=0$. Given a potential
$u \in H^2(\mathbb{R}) \cap H^{1,1}(\mathbb{R})$, it is true that
\begin{subequations}
\begin{eqnarray}
\label{equival1}
u^{(1)}(x) & = & \mathbf{B}_{\lambda_1}(\varphi_-(x;\lambda_1) e^{-i\lambda_1^2 x}) u(x) \\
\label{equival11}
& = & \mathbf{B}_{\lambda_1}(\phi_+(x;\lambda_1)e^{i\lambda_1^2 x}) u(x) \\
\label{equival2}
& = & \mathbf{B}_{\overline{\lambda}_1}(\varphi_+(x;\overline{\lambda}_1)e^{-i \overline{\lambda}_1^2x})u(x) \\
\label{equival22}
& = & \mathbf{B}_{\overline{\lambda}_1}(\phi_-(x;\overline{\lambda}_1)e^{i \overline{\lambda}_1^2x}) u(x),
\end{eqnarray}
\end{subequations}
where the four Jost functions to the KN spectral problem \eqref{laxeq1} are given
in Propositions \ref{KN-existence} and \ref{prop-Jost-KN}.
\end{lem}

\begin{proof}
Representation (\ref{equival1}) was defined in Lemma \ref{BT-lemma1}.
If $a(\lambda_1) = 0$, the representation (\ref{equival11}) was also
obtained in Lemma \ref{BT-lemma1}, thanks to the invariance of $G_{\lambda}$ and $S_{\lambda}$
under a multiplication by a nonzero complex number $a$
in properties (\ref{Gproperty1}) and (\ref{Sproperty1}) of Appendix \ref{algebra}
and the relation (\ref{norming-coeff}) between $\varphi_-(x;\lambda_1) e^{-i \lambda_1^2 x}$
and $\phi_+(x;\lambda_1) e^{i \lambda_1^2 x}$.
In order to establish (\ref{equival2}), we use the symmetry relation (\ref{symmetry-relation})
as well as properties \eqref{Gproperty2} and \eqref{Sproperty2} of Appendix \ref{algebra} and obtain
$$
G_{\lambda_1}(\varphi_-(x;\lambda_1)) = G_{\overline{\lambda}_1}(\sigma_1\sigma_3\varphi_-(x;\lambda_1)) =
G_{\overline{\lambda}_1}\left(\overline{\phi_-(x;\overline{\lambda}_-)}\right) = G_{\overline{\lambda}_1}(\phi_-(x;\overline{\lambda}_-))
$$
and
$$
S_{\lambda_1}(\varphi_-(x;\lambda_1))=-\overline{S_{\lambda_1}(\sigma_1\sigma_3\varphi_-(x;\lambda_1))} =
-\overline{S_{\lambda_1}\left(\overline{\phi_-(x;\overline{\lambda}_1)}\right)} = S_{\overline{\lambda}_1}(\phi_-(x;\overline{\lambda}_1)).
$$
The transformation formula (\ref{backlund}) yields (\ref{equival2}).
Finally, the representation (\ref{equival22}) is obtained from the relation
between $\varphi_+(x;\overline{\lambda}_1)e^{-i \overline{\lambda}_1^2 x}$ and
$\phi_-(x;\overline{\lambda}_1) e^{i \overline{\lambda}_1^2 x}$ in the case $a(\lambda_1) = 0$
that corresponds to $\overline{a(\overline{\lambda}_1)} = 0$.
\end{proof}

\subsection{Transformation of Jost functions}

For values of $\lambda \in \C\setminus\{\pm\lambda_1\}$, Jost functions of the KN spectral problem
(\ref{laxeq1})  associated with the new potential $u^{(1)} = \mathbf{B}_{\lambda_1}(\eta)u$
can be constructed from the old Jost functions by using the transformation matrix
\begin{equation}\label{B-def}
M[\eta, \lambda,\lambda_1] :=\frac{\lambda_1}{\overline{\lambda}_1}\frac{1}{\lambda^2-\lambda_1^2}
\left[\begin{matrix} \lambda^2G_{\lambda_1}(\eta)-|\lambda_1|^2 & -\frac{\lambda}{2i} S_{\lambda_1}(\eta) \\
\frac{\lambda}{2i} \overline{S_{\lambda_1}(\eta)} & -\lambda^2\overline{G_{\lambda_1}(\eta)}+|\lambda_1|^2
\end{matrix}\right].
\end{equation}
The following lemma presents the new Jost functions
of the KN spectral problem (\ref{laxeq1}) associated with the new potential $u^{(1)}$.
Since $u^{(1)} \in H^2(\mathbb{R}) \cap H^{1,1}(\mathbb{R})$ by Lemma \ref{BT-lemma1},
the new Jost functions exist according to Proposition \ref{KN-existence}.

\begin{lem}
\label{lem-new-Jost}
Under the same assumption as in Lemma \ref{BT-lemma1},  let us define for
$\lambda \in \C \setminus\{\pm\lambda_1,\pm \overline{\lambda}_1\}$,
\begin{subequations}
\begin{eqnarray}
\label{new-jost-1}
\varphi_-^{(1)}(x;\lambda) & = & M[\varphi_-(x;\lambda_1)e^{-i\lambda_1^2 x},\lambda,\lambda_1]\varphi_-(x;\lambda), \\
\label{new-jost-2}
\varphi_+^{(1)}(x;\lambda) & = & M[\varphi_+(x;\overline{\lambda}_1)e^{-i \overline{\lambda}_1^2x} ,\lambda,\overline{\lambda}_1]\varphi_+(x;\lambda),\\
\label{new-jost-3}
\phi_+^{(1)}(x;\lambda) & = & -M[\phi_+(x;\lambda_1)e^{i\lambda_1^2 x},\lambda,\lambda_1]\phi_+(x;\lambda), \\
\label{new-jost-4}
\phi_-^{(1)}(x;\lambda) & = & -M[\phi_-(x;\overline{\lambda}_1)e^{i \overline{\lambda}_1^2x} ,\lambda,\overline{\lambda}_1]\phi_-(x;\lambda).
\end{eqnarray}
\end{subequations}
Then, $\{\varphi_{\pm}^{(1)}(x;\lambda)e^{-i\lambda^2x}, \phi_{\pm}^{(1)}(x;\lambda)e^{i\lambda^2x}\}$
are Jost functions of the KN spectral problem (\ref{laxeq1})
associated with the potential $u^{(1)} = \mathbf{B}_{\lambda_1}(\eta)u$.
\end{lem}

\begin{proof}
First, we prove that the transformations (\ref{new-jost-1})--(\ref{new-jost-4})
produce solutions of the KN spectral problem associated with the potential $u^{(1)}$.
It is sufficient to consider the first Jost function in (\ref{new-jost-1}).
Therefore we shall verify that
\begin{equation}
\label{eq-varphi-one}
\partial_x(\varphi^{(1)}_-(x;\lambda)e^{-i\lambda^2 x})=\left[ -i\lambda^2\sigma_3+\lambda Q(u^{(1)}) \right] \varphi^{(1)}_-(x;\lambda)e^{-i\lambda^2 x}.
\end{equation}
Denoting entries of $M[\varphi_-(x;\lambda_1)e^{-i\lambda_1^2 x},\lambda,\lambda_1]$ by $M_{ij}$ for $1 \leq i,j \leq 2$
and using the KN spectral problem (\ref{laxeq1}) for $\varphi_-(x;\lambda) e^{-i \lambda^2 x}$,
we obtain the four differential equations:
\begin{subequations}\label{new-jost-proof-1}
\begin{align}
    &\partial_x M_{11}-\lambda\overline{u}M_{12} = \lambda u^{(1)}M_{21}&\label{new-jost-proof-1a} \\
    &\partial_x M_{12}+\lambda u M_{11} = \lambda u^{(1)}M_{22} - 2i\lambda^2M_{12} & \label{new-jost-proof-1b}\\
    &\partial_x M_{21}-\lambda\overline{u}M_{22} = -\lambda\overline{u}^{(1)}M_{11}+2i\lambda^2 M_{21}& \label{new-jost-proof-1c}\\
    &\partial_x M_{22}+\lambda u M_{21} = -\lambda\overline{u}^{(1)}M_{12} &\label{new-jost-proof-1d}
\end{align}
\end{subequations}
By using equation (\ref{d_lambda derivative}), we obtain
\begin{equation*}
    \partial_x G_{\lambda_1}(\eta) = \frac{\lambda_1^2-\overline{\lambda}_1^2} {[d_{\lambda_1}(\eta,\eta)]^2} \left[2i(\lambda_1^2-\overline{\lambda}_1^2) |\eta_1|^2|\eta_2|^2- d_{\lambda_1}(\eta,\eta) \overline{u}\eta_1\overline{\eta}_2- d_{\overline{\lambda}_1}(\eta,\eta) u\overline{\eta}_1\eta_2\right],
\end{equation*}
from which we verify equation (\ref{new-jost-proof-1a}) as follows:
\begin{eqnarray*}
  \partial_x M_{11}-\lambda\overline{u}M_{12}\!\!\!\! &\!\!=\!\!&\!\!
  \frac{\lambda_1}{\overline{\lambda}_1} \frac{\lambda^2}{\lambda^2-\lambda_1^2}
  \frac{\lambda_1^2-\overline{\lambda}_1^2} {[d_{\lambda_1}(\eta,\eta)]^2}
  \left[2i(\lambda_1^2-\overline{\lambda}_1^2) |\eta_1|^2|\eta_2|^2- d_{\overline{\lambda}_1}(\eta,\eta) u\overline{\eta}_1\eta_2\right]\\ &=& \lambda u^{(1)}M_{21}.
\end{eqnarray*}
The proof of (\ref{new-jost-proof-1d}) is based on the complex conjugate
equation and similar computations.

Equation (\ref{new-jost-proof-1b}) is equivalent to
$$
\partial_x (S_{\lambda_1}(\eta))= -2i( u + u^{(1)}) |\lambda_1|^2.
$$
This equality holds by means of the following two explicit computations:
\begin{equation*}
  \partial_x S_{\lambda_1}(\eta)= \frac{2i(\lambda_1^2-\overline{\lambda}_1^2)} {[d_{\lambda_1}(\eta,\eta)]^2} \left[-u|\lambda_1|^2 (|\eta_1|^4-|\eta_2|^4)-2i|\lambda_1|^2 \eta_1\overline{\eta}_2d_{\overline{\lambda}_1} (\eta,\eta)\right]
\end{equation*}
and
\begin{equation*}
  u^{(1)}+u=\frac{u(\lambda_1^2-\overline{\lambda}_1^2) (|\eta_1|^4-|\eta_2|^4)} {[d_{\lambda_1}(\eta,\eta)]^2}+ \frac{2i(\lambda_1^2-\overline{\lambda}_1^2) \eta_1\overline{\eta}_2d_{\overline{\lambda}_1}(\eta,\eta)} {[d_{\lambda_1}(\eta,\eta)]^2}.
\end{equation*}
Hence, we have proven (\ref{new-jost-proof-1b}). Equation (\ref{new-jost-proof-1c})
is obtained from complex conjugate equations and similar computations.

Thus, the function $\varphi^{(1)}_-(x;\lambda)e^{-i\lambda^2 x}$
satisfies equation (\ref{eq-varphi-one}), that is, it is a solution of the KN spectral problem (\ref{laxeq1})
associated with the potential $u^{(1)} = \mathbf{B}_{\lambda_1}(\eta)u$.
Similar computations are performed for the other functions  $\varphi_+^{(1)}(x;\lambda)e^{-i\lambda^2x}$ and
$\phi_{\pm}^{(1)}(x;\lambda)e^{i\lambda^2x}$
given by (\ref{new-jost-2})--(\ref{new-jost-4}).
Since $G_{\lambda_1}(\eta)$ and $S_{\lambda_1}(\eta)$ are bounded in $x$ for all considered choices for
$\eta$, the functions $\varphi_{\pm}^{(1)}(x;\lambda)$ and $\phi_{\pm}^{(1)}(x;\lambda)$
are bounded functions of $x$ for every $\lambda \in \C \setminus\{\pm\lambda_1,\pm \overline{\lambda}_1\}$.

It is left to check the boundary conditions
(\ref{jost-infinity-original}) in Proposition \ref{KN-existence}.
The boundary conditions (\ref{jost-infinity-original}) follow from
properties (\ref{G-comp}) and (\ref{G-comp-again}) in Appendix \ref{algebra}:
\begin{equation} \label{B-property}
M[e_1,\lambda,\lambda_1]e_1 = e_1,\quad M[e_2,\lambda,\lambda_1]e_2 =-e_2.
\end{equation}
Since $u^{(1)} \in H^2(\mathbb{R}) \cap H^{1,1}(\mathbb{R})$ satisfies the assumption
of Proposition \ref{KN-existence}, the four functions
(\ref{new-jost-1})--(\ref{new-jost-4}) are the unique Jost functions of the KN spectral problem (\ref{laxeq1}) associated with
$u^{(1)}$.
\end{proof}

Since the definitions (\ref{new-jost-1})--(\ref{new-jost-4}) with the transformation
matrix (\ref{B-def}) are singular as $\lambda \to \{ \pm \lambda_1, \pm \overline{\lambda}_1\}$,
we show that the singularity is removable, so that the definitions (\ref{new-jost-1})--(\ref{new-jost-4})
can be extended in the domains of analyticity of the Jost functions $\varphi_{\pm}(x;\lambda)$ and $\phi_{\pm}(x;\lambda)$
according to Proposition \ref{prop-Jost-KN}.

\begin{lem}
\label{lem-new-Jost-analyticity}
Let $\varphi_{\pm}^{(1)}(x;\lambda)$ and $\phi_{\pm}^{(1)}(x;\lambda)$ be defined by
(\ref{new-jost-1})--(\ref{new-jost-4}). Then, $\lambda = \pm \lambda_1$ and $\lambda = \pm \overline{\lambda}_1$
are removable singularities in the corresponding domains of analyticity of
$\varphi^{(1)}_{\pm}(x;\lambda)$ and $\phi^{(1)}_{\pm}(x;\lambda)$.
\end{lem}

\begin{proof}
It is sufficient again to consider the first Jost function $\varphi_-^{(1)}(x;\lambda)$
represented by (\ref{new-jost-1}). By using the notations
$\varphi_-=(\varphi_{-,1}, \varphi_{-,2})^t$ and $\varphi^{(1)}_-=(\varphi^{(1)}_{-,1}, \varphi^{(1)}_{-,2})^t$
for the $2$-vectors and dropping the dependence on $x$, we obtain for $\lambda\in\C_{I} \cup \C_{III} \setminus\{\pm \lambda_1\}$
\begin{eqnarray*}
  \varphi^{(1)}_{-,1}(\lambda) &=&  \frac{\lambda_1}{\overline{\lambda}_1} \left\{\frac{\left(\lambda^2 d_{\overline{\lambda}_1} (\varphi_-,\varphi_-) -|\lambda_1|^2 d_{\lambda_1} (\varphi_-,\varphi_-)\right) \varphi_{-,1}(\lambda)}{(\lambda^2-\lambda_1^2)\: d_{\lambda_1} (\varphi_-,\varphi_-)}\right. \\
  &&\qquad\qquad\qquad-\left.\frac{\lambda (\lambda_1^2-\overline{\lambda}_1^2) \varphi_{-,1}(\lambda_1) \overline{\varphi_{-,2}(\lambda_1)} \varphi_{-,2}(\lambda)}{(\lambda^2-\lambda_1^2)\: d_{\lambda_1} (\varphi_-,\varphi_-)}\right\}\\
   &=& \frac{\lambda_1}{\overline{\lambda}_1} \frac{(\lambda^2-\lambda_1^2)\overline{\lambda}_1 |\varphi_{-,1}(\lambda_1)|^2\varphi_{-,1}(\lambda) +F(\lambda)}{(\lambda^2-\lambda_1^2)\: d_{\lambda_1} (\varphi_-,\varphi_-)},
\end{eqnarray*}
where
\begin{equation*}
  F(\lambda):=(\lambda^2-\overline{\lambda}_1^2) \lambda_1|\varphi_{-,2}(\lambda_1)|^2 \varphi_{-,1}(\lambda)-\lambda (\lambda_1^2-\overline{\lambda}_1^2) \varphi_{-,1}(\lambda_1) \overline{\varphi_{-,2}(\lambda_1)} \varphi_{-,2}(\lambda).
\end{equation*}
Since $\varphi_{-,1}(\lambda)$ is even in $\lambda$ and $\varphi_{-,2}(\lambda)$ is odd in $\lambda$ \cite{Pelinovsky-Shimabukuro},
we obviously have $F(\lambda_1) = F(-\lambda_1) = 0$. Furthermore, $F$ is analytic in
 $\mathbb{C}_I \cup \mathbb{C}_{III}$ by Proposition \ref{prop-Jost-KN}, hence
 $F(\lambda)=(\lambda^2-\lambda_1^2)\widetilde{F}(\lambda)$, where $\widetilde{F}$ is analytic in $\mathbb{C}_I \cup \mathbb{C}_{III}$.
Thus, we obtain
\begin{equation*}
  \varphi^{(1)}_{-,1}(\lambda)= \frac{\lambda_1}{\overline{\lambda}_1}
  \frac{\overline{\lambda}_1 |\varphi_{-,1}(\lambda_1)|^2 \varphi_{-,1}(\lambda) + \widetilde{F}(\lambda)}{d_{\lambda_1}(\varphi_-,\varphi_-)} ,
\end{equation*}
so that  $\pm \lambda_1$ are removable singularities of $\varphi^{(1)}_{-,1}(\lambda)$. Similar
calculations show that $\pm \lambda_1$ are also removable singularities of $\varphi^{(1)}_{-,2}(\lambda)$.
\end{proof}

\subsection{Transformation of scattering coefficients}

We next transform the scattering coefficients $a(\lambda)$ and $b(\lambda)$ given by (\ref{a})--(\ref{b})
and show that the new potential $u^{(1)}$ belongs to $Z_0 \subset H^2(\mathbb{R}) \cap H^{1,1}(\mathbb{R})$
if the old potential $u$ belongs to $Z_1 \subset H^2(\mathbb{R}) \cap H^{1,1}(\mathbb{R})$
and the value $\lambda_1 \in \C_I$ is chosen to be the root of $a(\lambda)$ in $\C_I$.
The following lemma gives the corresponding result.

\begin{lem} \label{nonzero-a}
Let $u\in Z_1$ and $\lambda_1 \in \C_I$ such that $a(\lambda_1)=0$. Let $\eta(x) = \varphi_-(x;\lambda_1) e^{-i\lambda_1^2 x}$,
where $\varphi_-$ is the Jost function of
the KN spectral problem \eqref{laxeq1} given in Propositions \ref{KN-existence} and \ref{prop-Jost-KN}.
Then, $u^{(1)}=\mathbf{B}_{\lambda_1}(\eta)u$ belongs to $Z_0$.
\end{lem}

\begin{proof}
In order to show that $u^{(1)} \in Z_0$, we show that if the only simple zero
$a(\lambda) = W(\varphi_-(\cdot;\lambda), \phi_+(\cdot;\lambda))$
in $\mathbb{C}_I$ is located at $\lambda = \lambda_1$, then
$a^{(1)}(\lambda) := W(\varphi_-^{(1)}(\cdot;\lambda), \phi_+^{(1)}(\cdot;\lambda))$ has no zero in $\C_I$,
where $\varphi_-^{(1)}$ and $\phi_+^{(1)}$ are given by (\ref{new-jost-1}) and (\ref{new-jost-3}) in Lemma \ref{lem-new-Jost}.
This follows from the direct computation as follows:
\begin{subequations}
\begin{eqnarray}
 &a^{(1)}(\lambda)&= W(\varphi_-^{(1)}(x;\lambda), \phi_+^{(1)}(x;\lambda)) \label{a1} \\
&&=W\Big(M[\varphi_-(x;\lambda_1) e^{-i \lambda_1^2 x},\lambda,\lambda_1]\varphi_-(x;\lambda), \label{a2} \\&&\qquad\qquad-M[\phi_+(x;\lambda_1) e^{i \lambda_1^2 x},\lambda,\lambda_1]\phi_+(x;\lambda)\Big)\nonumber
 \\
&&=W\Big(M[\varphi_-(x;\lambda_1),\lambda,\lambda_1] \varphi_-(x;\lambda), \label{a3} \\ &&\qquad\qquad -M[\varphi_-(x;\lambda_1),\lambda,\lambda_1]\phi_+(x;\lambda)\Big) \nonumber \\
&&=-a(\lambda) \, \det\left(M[\varphi_-(x;\lambda_1) ,\lambda,\lambda_1]\right) \label{a4} \\
&&= -a(\lambda) \det\left(M[e_1 ,\lambda,\lambda_1]\right) \label{a5} \\
&&=a(\lambda)\frac{\lambda_1^2}{\overline {\lambda}_1^2}\frac{\lambda^2-\overline{\lambda}_1^2} {\lambda^2-\lambda_1^2},\label{a6}
\end{eqnarray}
\end{subequations}
where we have used (\ref{new-jost-1}) and (\ref{new-jost-3}) to get (\ref{a2}),
(\ref{norming-coeff}), (\ref{Gproperty1}), and (\ref{Sproperty1})
to get (\ref{a3}), (\ref{a}) to get (\ref{a4}),
the limit $x \to -\infty$ to get (\ref{a5}),
and (\ref{Sproperty1}), (\ref{G-comp}) and (\ref{G-comp-again}) to get (\ref{a6}).
Since $\lambda_1$ is the only simple zero of $a(\lambda)$ in $\mathbb{C}_I$,
then $a^{(1)}(\lambda)$ has no zeros for $\lambda$ in $\mathbb{C}_I$.
\end{proof}

\begin{remark}\label{new-b}
For completeness, we also give transformation of $b(\lambda)$ to $b^{(1)}(\lambda)$ as follows:
\begin{align*}
b^{(1)}(\lambda) &=W(e^{-2i \lambda^2 x} \varphi_+^{(1)}(x;\lambda), \varphi_-^{(1)}(x;\lambda)) \\
& = W(e^{-2i \lambda^2 x} \varphi_+^{(1)}(x;\lambda),M[\varphi_-(x;\lambda_1) e^{-i \lambda_1^2 x},\lambda,\lambda_1]\varphi_-(x;\lambda)) \\
& = W\Big(e^{-2i \lambda^2 x} \varphi_+^{(1)}(x;\lambda),\\&\qquad \qquad M[\phi_+(x;\lambda_1) e^{i \lambda_1^2 x},\lambda,\lambda_1] [a(\lambda)\varphi_+(x;\lambda) + e^{i2\lambda^2x} b(\lambda) \phi_+(x;\lambda)]\Big) \\
& = b(\lambda) W(e_1,M[e_2,\lambda,\lambda_1]e_2)\\
&=-b(\lambda),
\end{align*}
where the term with $a(\lambda)$ vanishes in the limit $x \to +\infty$
because $W(e_1,e_1) = 0$ and we have used the following limits as $x \to +\infty$
$$
M[\phi_+(x;\lambda_1),\lambda,\lambda_1]\varphi_+(x;\lambda)
\rightarrow M[e_2,\lambda,\lambda_1] e_1= \frac{\lambda_1^2 (\lambda^2-\overline{\lambda}_1^2)}{\overline{\lambda}_1^2 (\lambda^2 - \lambda_1^2)} e_1
$$
and $\varphi_+^{(1)}(x;\lambda) \rightarrow e_1$.
\end{remark}

By Lemmas \ref{BT-lemma1}, \ref{pull-back-eta}, and \ref{nonzero-a},
we have shown the existence of a sequence of invertible B\"{a}cklund transformations $Z_1\rightarrow Z_0 \rightarrow Z_1$ given by
$$
u \rightarrow \mathbf{B}_{\lambda_1}(\eta) \rightarrow u^{(1)} \rightarrow \mathbf{B}_{\lambda_1}(\eta^{(1)}) \rightarrow u.
$$
Next, we express $\eta^{(1)}$ in Lemma \ref{pull-back-eta}
in terms of the new Jost functions $\varphi_{\pm}^{(1)}$ and $\phi_{\pm}^{(1)}$
associated with $u^{(1)}$ in Lemma \ref{lem-new-Jost}.

\begin{lem}\label{eta-by-new-jost}
Fix $\lambda_1 \in \C_I$ such that $a(\lambda_1)=0$ and $a'(\lambda_1) \neq 0$. Let $\eta$ and $\eta^{(1)}$ be given as in Lemmas
\ref{BT-lemma1} and \ref{pull-back-eta}. Then, $\eta^{(1)}$ is decomposed as
\begin{equation}\label{eta-decomp}
\eta^{(1)}(x) = \frac{1}{\gamma_1\overline{\lambda}_1 a^{(1)}(\lambda_1)}e^{-i\lambda_1^2x}\varphi_-^{(1)}(x;\lambda_1)
+\frac{1}{\overline{\lambda}_1 a^{(1)}(\lambda_1)}e^{i\lambda_1^2x}\phi_+^{(1)}(x;\lambda_1),
\end{equation}
where the new Jost functions $\varphi_-^{(1)}$ and $\phi_+^{(1)}$ are constructed in Lemmas \ref{lem-new-Jost} and \ref{lem-new-Jost-analyticity},
$\gamma_1 \neq 0$ is the norming constant in (\ref{norming-coeff}), and $a^{(1)}(\lambda_1) \neq 0$ as in Lemma \ref{nonzero-a}.
\end{lem}

\begin{proof}
We use notations $\eta^{(1)}=(\eta^{(1)}_1, \eta^{(1)}_2)^t$ and $\varphi_-=(\varphi_{-,1}, \varphi_{-,2})^t$
for the $2$-vectors. Components of $\eta^{(1)}$ given by (\ref{new-eig1}) are rewritten explicitly by
\begin{equation*} 
\eta_1^{(1)}(x) = \frac{e^{i\lambda_1^2x} \overline{\varphi_{-,2}(x;\lambda_1)}}{
d_{\lambda_1}(\varphi_-(x;\lambda_1),\varphi_-(x;\lambda_1))}
\end{equation*}
and
\begin{equation*}
\eta_2^{(1)}(x) = \frac{e^{i\lambda_1^2x} \overline{\varphi_{-,1}(x;\lambda_1)}}{
d_{\overline{\lambda}_1}(\varphi_-(x;\lambda_1),\varphi_-(x;\lambda_1))}.
\end{equation*}
Since $\lim\limits_{x\rightarrow -\infty}\varphi_{-}(x;\lambda_1)=e_1$,
we have
\begin{equation} \label{eta-limit1}
\lim_{x\rightarrow -\infty}e^{-i\lambda_1^2x}\eta^{(1)}(x) = \frac{1}{d_{\overline{\lambda}_1}(e_1,e_1)} e_2 =
\frac{1}{\overline{\lambda}_1} e_2.
\end{equation}

By using the relation (\ref{norming-coeff}) with the norming coefficient $\gamma_1$, components of
$\eta^{(1)}$ can be rewritten in the equivalent form:
\begin{equation*} 
\eta^{(1)}_1(x) = \frac{e^{-i\lambda_1^2x}\overline{\phi_{+,2}(x;\lambda_1)}}{\gamma_1 d_{\lambda_1}(\phi_+(x;\lambda_1),\phi_+(x;\lambda_1))}
\end{equation*}
and
\begin{equation*}
\eta^{(1)}_{2}(x) = \frac{e^{-i\lambda_1^2x} \overline{\phi_{+,1}(x;\lambda_1)}}{\gamma_1 d_{\overline{\lambda}_1}(\phi_+(x;\lambda_1),\phi_+(x;\lambda_1))}.
\end{equation*}
Since $\lim\limits_{x\rightarrow +\infty}\phi_{+}(x;\lambda_1)=e_2$,
we have
\begin{equation} \label{eta-limit2}
\lim_{x\rightarrow +\infty}e^{i\lambda_1^2x}\eta^{(1)}(x) = \frac{1}{\gamma_1 d_{\lambda_1}(e_2,e_2)} e_1 =
\frac{1}{\gamma_1\overline{\lambda}_1}e_1.
\end{equation}

By Lemma \ref{pull-back-eta}, $\eta^{(1)}$ is a solution of the KN spectral problem (\ref{laxeq1})
associated with the new potential $u^{(1)}$ for $\lambda = \lambda_1$. By Lemmas \ref{lem-new-Jost}
and \ref{lem-new-Jost-analyticity}, the two new Jost functions
$\varphi_-^{(1)}(x;\lambda)$ and $\phi_+^{(1)}(x;\lambda)$ are analytic at $\lambda_1$.
Any solution of the second-order system is a linear combination of the two linearly independent
solutions, so that we have
$$
\eta^{(1)}(x) = c_1 \varphi_-^{(1)}(x;\lambda_1) e^{-i \lambda_1^2 x} + c_2 \phi_+^{(1)}(x;\lambda_1) e^{i \lambda_1^2 x},
$$
where $c_1$, $c_2$ are some numerical coefficients.
Thanks to the boundary conditions (\ref{jost-infinity-original}) and the representation (\ref{a}), we obtain the boundary conditions
\begin{equation}
\label{eta-limit3}
\lim_{x\rightarrow -\infty} e^{-i\lambda_1^2x}\eta^{(1)}(x) = c_2 a^{(1)}(\lambda_1) e_2, \quad
\lim_{x\rightarrow +\infty} e^{i\lambda_1^2x}\eta^{(1)}(x) = c_1 a^{(1)}(\lambda_1) e_1,
\end{equation}
where we have recalled that $\lambda_1 \in \C_I$. Since $a^{(1)}(\lambda_1) \neq 0$ by
Lemma \ref{nonzero-a}, $c_1$ and $c_2$ are found uniquely from \eqref{eta-limit1}, \eqref{eta-limit2},
and \eqref{eta-limit3} to yield the decomposition (\ref{eta-decomp}).
\end{proof}

\begin{remark}
Instead of the decomposition (\ref{eta-decomp}), we can write
\begin{equation}
\eta^{(1)}(x) := \varphi^{(1)}_-(x;\lambda_1) e^{-i\lambda_1^2 x}+\gamma_1\,\phi^{(1)}_+(x;\lambda_1) e^{i\lambda_1^2 x}
\end{equation}
because the B\"{a}cklund transformation (\ref{backlund}) is invariant if $\eta^{(1)}$ is multiplied by a nonzero constant.
\end{remark}

\begin{lem}
\label{inverse BT-lemma1}
Under the same conditions as in Lemma \ref{eta-by-new-jost}, for every $u^{(1)} \in H^2(\mathbb{R}) \cap H^{1,1}(\mathbb{R})$
satisfying $\|u^{(1)}\|_{H^2 \cap H^{1,1}}\leq M$ for some $M > 0$, the transformation
$$\mathbf{B}_{\lambda_1}(\eta^{(1)})u^{(1)} \in H^2(\mathbb{R}) \cap H^{1,1}(\mathbb{R})$$ satisfies
\begin{equation}
\label{estimate-key-inverse-backlund}
\| \mathbf{B}_{\lambda_1}(\eta^{(1)}) u^{(1)} \|_{H^2 \cap H^{1,1}} \leq C_M,
\end{equation}
where the constant $C_M$ does not depend on $u^{(1)}$.
\end{lem}

\begin{proof}
By the representation (\ref{a}), we have
\begin{eqnarray*}
  |a^{(1)}(\lambda_1)| &=&\left|\left(\varphi^{(1)}_{-,1}(x;\lambda_1) e^{-i\lambda_1^2 x}+\gamma_1\,\phi^{(1)}_{+,1}(x;\lambda_1)e^{i\lambda_1^2 x}\right) \phi^{(1)}_{+,2}(x;\lambda_1)e^{i\lambda_1^2 x}\phantom{\sum}\right.\\
  &&\left.\phantom{\sum} - \left(\varphi^{(1)}_{-,2}(x;\lambda_1)e^{-i\lambda_1^2 x}+ \gamma_1\,\phi^{(1)}_{+,2}(x;\lambda_1)e^{i\lambda_1^2 x}\right) \phi^{(1)}_{+,1}(x;\lambda_1)e^{i\lambda_1^2 x}\right|\\
  &\leq& \|\phi^{(1)}_+(\cdot;\lambda_1)\|_{L^{\infty}} \left( |e^{i\lambda_1^2 x}\eta^{(1)}_1(x)| + |e^{i\lambda_1^2 x}\eta^{(1)}_2(x)| \right).
\end{eqnarray*}
Since $a^{(1)}(\lambda_1)\neq 0$ by Lemma \ref{nonzero-a}
and $|d_{\lambda}(\eta,\eta)| \geq |{\rm Re}(\lambda_1)| (|\eta_1|^2 + |\eta_2|^2)$,
it follows from (\ref{KN L^infty estimate}) that
there is a constant $C_M > 0$ independently of $u^{(1)}$ such that
\begin{equation*}
  \frac{1}{|d_{\lambda_1}(e^{i\lambda_1^2 x}\eta^{(1)}(x),e^{i\lambda_1^2 x}\eta^{(1)}(x))|}\leq
  C_M \quad \text{for all }x\in\R.
\end{equation*}
By using the same argument, we also obtain
\begin{equation*}
  |a^{(1)}(\lambda_1)|\leq |\gamma_1|^{-1} \|\varphi^{(1)}_-(\cdot;\lambda_1) \|_{L^{\infty}}
  \left( |e^{-i\lambda_1^2 x}\eta_1^{(1)}(x)| + |e^{-i\lambda_1^2 x}\eta_2^{(1)}(x)|\right),
\end{equation*}
such that
\begin{equation*}
  \frac{1}{|d_{\lambda_1}(e^{-i\lambda_1^2 x}\eta^{(1)}(x),e^{-i\lambda_1^2 x}\eta^{(1)}(x))|}\leq
  C_M\quad \text{for all }x\in\R.
\end{equation*}
As a consequence, by using the bound
\begin{eqnarray*}
  \left|\frac{\eta^{(1)}_1\overline{\eta}^{(1)}_2} {d_{\lambda_1}(\eta^{(1)},\eta^{(1)})}\right| &\leq& \frac{|\varphi^{(1)}_{-,1}(x;\lambda_1) \overline{\varphi^{(1)}_{-,2} (x;\lambda_1)}|} {|d_{\lambda_1}(e^{i \lambda_1^2 x}\eta^{(1)},e^{i \lambda_1^2 x}\eta^{(1)})|}+|\gamma_1|^2\frac{|\phi^{(1)}_{-,1}(x;\lambda_1) \overline{\phi^{(1)}_{-,2} (x;\lambda_1)}|} {|d_{\lambda_1}(e^{-i \lambda_1^2 x}\eta^{(1)},e^{-i \lambda_1^2 x}\eta^{(1)})|} \\
   && \qquad+|\gamma_1|\frac{|\varphi^{(1)}_{-,1}(x;\lambda_1) \overline{\phi^{(1)}_{-,2} (x;\lambda_1)}|+|\phi^{(1)}_{-,1}(x;\lambda_1) \overline{\varphi^{(1)}_{-,2} (x;\lambda_1)}|}{|d_{\lambda_1}(\eta^{(1)},\eta^{(1)})|},
\end{eqnarray*}
and the bounds (\ref{norm1})--(\ref{norm2}) of Appendix B, we obtain
$$
\|S_{\lambda_1}(\eta^{(1)})u^{(1)}\|_{L^{2,1}}\leq C_M.
$$
Similar to the proof of Lemma \ref{BT-lemma1}, this implies by the triangle inequality that
$$
\| \mathbf{B}_{\lambda_1}(\eta^{(1)}) u^{(1)} \|_{L^{2,1}} \leq C_M.
$$
The norms $\| \partial_x(\mathbf{B}_{\lambda_1}(\eta^{(1)}) u^{(1)} )\|_{L^{2,1}}$ and
$\| \partial_x^2(\mathbf{B}_{\lambda_1}(\eta^{(1)}) u^{(1)} )\|_{L^{2}} $ can be estimated
similarly with the use of estimates (\ref{norm1})--(\ref{norm2}) of Appendix B, so that
the proof of the bound (\ref{estimate-key-inverse-backlund}) is complete.
\end{proof}

\section{Time evolution of the B\"{a}cklund transformation}

Here we will prove property (iv) claimed in Section 3. In other words,
extending the Jost function $\varphi_-(t,x;\lambda)$ to be time-dependent
according to the linear system (\ref{laxeq1}) and (\ref{laxeq2}),
we will prove the following lemma, which is a time-dependent analogue
of Lemma \ref{BT-lemma1}.

\begin{lem}\label{same_solutions}
Fix $\lambda_1 \in \mathbb{C}_I$. Given a local solution $u(t,\cdot) \in H^2(\mathbb{R}) \cap H^{1,1}(\mathbb{R})$,
$t \in (-T,T)$ to the Cauchy problem (\ref{dnls}) for some $T > 0$,
define $$\eta(t,x) := \varphi_-(t,x;\lambda_1) e^{-i(\lambda_1^2 x+2 \lambda_1^4 t)},$$
where $\varphi_-$ is the Jost function of the linear system \eqref{laxeq1} and \eqref{laxeq2}.
Then, $u^{(1)}(t,\cdot) = \mathbf{B}_{\lambda_1}(\eta(t,\cdot))u(t,\cdot)$
belongs to $H^2(\mathbb{R}) \cap H^{1,1}(\mathbb{R})$ for every $t \in [0,T)$
and satisfies the Cauchy problem (\ref{dnls}) for $u^{(1)}(0,\cdot) = \mathbf{B}_{\lambda_1}(\eta(0,\cdot))u(0,\cdot)$.
\end{lem}

One way to prove Lemma \ref{same_solutions} is to show that
the time-dependent versions of the transformations (\ref{new-jost-1})--(\ref{new-jost-4})
satisfy the time evolution equation (\ref{laxeq2}) associated with the
potential $u^{(1)}(t,\cdot) = \mathbf{B}_{\lambda_1}(\eta(t,\cdot)) u(t,\cdot)$.
By compatibility of the linear system
(\ref{laxeq1}) and (\ref{laxeq2}) as well as smoothness of the new Jost functions and the new potential $u^{(1)}$,
it then follows that $u^{(1)}(t,x)$ is a new solution of the DNLS equation $i u_t + u_{xx} + i (|u|^2 u)_x = 0$.

However, the proof of the above claim is straightforward but enormously lengthy.
Therefore we will avoid the technical proof and instead make use of the
inverse scattering transform for the soliton-free solutions to the Cauchy problem (\ref{dnls}),
which was developed in
the recent works \cite{Perry,Pelinovsky-Shimabukuro}. We explain the idea
for the case of one soliton and then extend the argument to the case of finitely many solitons.

Let $u(t,\cdot) \in Z_1 \subset H^2(\mathbb{R}) \cap H^{1,1}(\mathbb{R})$ be a local solution of
the Cauchy problem (\ref{dnls}) on $(-T,T)$ for some $T > 0$. For every fixed time $t \in (-T,T)$
we find a new potential of the KN spectral problem (\ref{laxeq1}) by means of the
B\"{a}cklund transformation $u^{(1)}(t,\cdot) = \mathbf{B}_{\lambda_1}(\eta(t,\cdot)) u(t,\cdot)$.
If $\lambda_1 \in \C_I$ is taken such that $a(\lambda_1) = 0$, then $u^{(1)}(t,\cdot) \in Z_0$.
On the other hand, let $\widetilde{u}(t,\cdot) \in Z_0$ be a solution to the Cauchy problem
(\ref{dnls}) starting with the initial condition $\widetilde{u}(0,\cdot) = u^{(1)}(0,\cdot) \in Z_0$.
Since assumptions of \cite[Theorem 1.1]{Pelinovsky-Shimabukuro} are satisfied,
the solution $\widetilde{u}(t,\cdot) \in Z_0$ exists for every $t \in \mathbb{R}$,
in particular, for $t \in (-T,T)$. The following diagram illustrates the scheme.
\begin{equation*}
\begin{xy}\xymatrixcolsep{3pc}
  \xymatrix{
    u(0,\cdot)\in Z_1\ar@{->}[d]_{\text{DNLS}} \ar@{->}[r]&  u^{(1)}(0,\cdot)\in Z_0 \ar@{->}[rd]^{\text{DNLS}}& &\\
    u(t,\cdot)\in Z_1  \ar@{->}[r] & u^{(1)}(t,\cdot)\in Z_0\ar@{<->}[r]^{?}&\widetilde{u}(t,\cdot)\in Z_0, & \!\!\!\!\!\!\!\!t \in (-T,T).
 }
\end{xy}
\end{equation*}
Thus, the proof of Lemma \ref{same_solutions} in the case of one soliton will rely on the
proof that $\widetilde{u}(t,\cdot) = u^{(1)}(t,\cdot)$ for every $t \in (-T,T)$.
To show this, we will first prove that the two functions have the same scattering data.

\begin{lem}\label{same scattering data}
  For every $t \in (-T,T)$, the potentials $\widetilde{u}(t,\cdot)$ and $u^{(1)}(t,\cdot)$ produce the same scattering data.
\end{lem}

\begin{proof}
  We know that both functions $\widetilde{u}(t,\cdot)$ and $u^{(1)}(t,\cdot)$
  remain in $Z_0$ for every $t \in (-T,T)$. Hence the scattering data consist
  only of the reflection coefficient which is introduced in \cite{Pelinovsky-Shimabukuro}.
  For the potential $u(t,\cdot)\in Z_1$ with $t \in (-T,T)$, we have $r(t,\lambda)=b(t,\lambda)/a(t,\lambda)$
  for $\lambda\in \R\cup i\R$. Let us denote by $r^{(1)}(t,\lambda)=b^{(1)}(t,\lambda)/a^{(1)}(t,\lambda)$ the reflection coefficient
  of $u^{(1)}(t,\cdot)\in Z_0$ for $t \in (-T,T)$. Lemma \ref{nonzero-a} and Remark \ref{new-b} tell us how
  the old and the new reflection coefficient are connected:
  \begin{equation}
  \label{transformation-r}
    r^{(1)}(t,\lambda)=-r(t,\lambda)\frac{\lambda_1^2}{\overline{\lambda}_1^2} \frac{\lambda^2-\overline{\lambda}_1^2}{\lambda^2-\lambda_1^2},
    \quad \lambda \in \mathbb{R} \cup i \mathbb{R}, \quad t \in (-T,T).
  \end{equation}
  When $u(t,x)$ is a solution to the DNLS equation in $Z_1$, the time evolution of the reflection coefficient $r(t,\lambda)$
  is given by
  \begin{equation}
  \label{time-r}
    r(t,\lambda)=r(0,\lambda)e^{4i\lambda^4t}, \quad t \in (-T,T),
  \end{equation}
which implies by virtue of (\ref{transformation-r}) that
  \begin{equation}
  \label{evolution-r}
    r^{(1)}(t,\lambda)=r^{(1)}(0,\lambda)e^{4i\lambda^4t}, \quad t \in (-T,T).
  \end{equation}
  Note that equation (\ref{time-r}) coincides with Eq. (5.2) in \cite{Pelinovsky-Shimabukuro}, where it was derived for $u$ in $Z_0$.
  The proof for $u$ in $Z_1$ is the same.

  For the reflection coefficient $\widetilde{r}$ of the potential $\widetilde{u}$ we know $r^{(1)}(0,\lambda)=\widetilde{r}(0,\lambda)$ since
  $u^{(1)}(0,\cdot) = \widetilde{u}(0,\cdot)$. By using the time evolution of the reflection coefficient
  from \cite{Pelinovsky-Shimabukuro} and the expression (\ref{evolution-r}), we obtain
  \begin{equation}
    \widetilde{r}(t,\lambda)=\widetilde{r}(0,\lambda) e^{4i\lambda^4t}=r^{(1)}(0,\lambda)e^{4i\lambda^4t}= r^{(1)}(t,\lambda), \quad t \in (-T,T).
  \end{equation}
  The assertion of the lemma is proved.
\end{proof}

\begin{cor}\label{corollary}
  The potential $u^{(1)}(t,\cdot) = \mathbf{B}_{\lambda_1}(\eta(t,\cdot)) u(t,\cdot)$ is a new solution of the DNLS equation for $t \in (-T,T)$.
\end{cor}

\begin{proof}
  In \cite{Perry,Pelinovsky-Shimabukuro}, existence and Lipschitz continuity of the mapping
  \begin{equation*}
  L^{2,1}(\mathbb{R} \cup i \mathbb{R}) \supset X \ni r\mapsto u \in Z_0 \subset H^2(\mathbb{R}) \cap H^{1,1}(\mathbb{R})
  \end{equation*}
  was established by means of the solvability
  of the associated Riemann--Hilbert problem (see \cite{Perry,Pelinovsky-Shimabukuro} for details on $X$). Therefore, the mapping is bijective and
  $\widetilde{u}(t,\cdot)=u^{(1)}(t,\cdot)$ for every $t \in (-T,T)$ follows from Lemma \ref{same scattering data}.
  Since $\widetilde{u}$ is a solution of the DNLS equation, so does $u^{(1)}$.
\end{proof}

The proof of Lemma \ref{same_solutions} in the case of finitely many solitons relies on
the iterative use of the argument above. For a given $u \in Z_k$, $k \in \mathbb{N}$,
we remove the distinct eigenvalues $\{\lambda_1,...\lambda_k\}$ in $\C_I$
by iterating the B\"{a}cklund transformation $k$ times. We set $u^{(0)} = u$ and
\begin{eqnarray*}
  u^{(l)} = \mathbf{B}_{\lambda_l}(\eta^{(l-1)})u^{(l-1)},\quad (1\leq l \leq k),
\end{eqnarray*}
such that eventually $u^{(k)}\in Z_0$. The arguments of Lemma \ref{same scattering data}
and Corollary \ref{corollary} apply to the last potential $u^{(k)}$. As a result,
we know that the $k$-fold iteration of the B\"{a}cklund transformation of a solution
$u(t,\cdot)\in Z_k$ of the Cauchy problem (\ref{dnls}) for $t \in [0,T)$
produces a new solution $u^{(k)}(t,\cdot) \in Z_0$ of the Cauchy problem (\ref{dnls}). Thus, the following diagram commutes.

\begin{equation*}
\begin{xy}\xymatrixcolsep{2pc}
  \xymatrix{
    u(0,\cdot)\in Z_k\ar@{->}[d]_{\text{DNLS}} \ar@{->}[r]& u^{(1)}(0,\cdot)\in Z_{k-1}\ar@{->}[r]&\cdots\ar@{->}[r]& u^{(k)}(0,\cdot)\in Z_0 \ar@{->}[d]^{\text{DNLS}}&\\
    u(t,\cdot)\in Z_k  \ar@{->}[r]&u^{(1)}(t,\cdot)\in Z_{k-1} \ar@{->}[r]& \cdots\ar@{->}[r]&u^{(k)}(t,\cdot)\in Z_0& \!\!\!\!\!\!\!
 }
\end{xy}
\end{equation*}

\begin{remark}
  We do not prove here that every step in the chain
  $
    u\to u^{(1)}\to\cdots\to u^{(k)}
  $
  yields a solution of the DNLS equation. Although this claim is likely to be true, the proof would require
  the inverse scattering theory of \cite{Perry,Pelinovsky-Shimabukuro}
  to be extended to the cases of eigenvalues.
\end{remark}

\section{An example of the B\"{a}cklund transformation}
Let us give an example of the explicit B\"{a}cklund transformation that connects the zero and
one-soliton solutions of the DNLS equation. In order to find the one-soliton solution in the explicit form,
we assume that we start with a potential $u_{\lambda_1,\gamma_1} \in Z_1$ with eigenvalue $\lambda_1\in\C_{I}$
and norming constant $\gamma_1\neq0$, for which the B\"{a}cklund transformation (\ref{backlund})
in Lemma \ref{BT-lemma1} yields exactly the zero solution:
\begin{equation}
\label{back-1}
u^{(1)} =  \mathbf{B}_{\lambda_1}(\eta)u_{\lambda_1,\gamma_1} = 0.
\end{equation}
For the zero solution, we know that the Jost functions of the linear system (\ref{laxeq1}) and (\ref{laxeq2}) are given by
$$
\varphi_{\pm}^{(1)}(t,x;\lambda) = e^{-i(\lambda^2x+2\lambda^4t)} e_1, \quad
\phi_{\pm}^{(1)}(t,x;\lambda) = e^{i(\lambda^2x+2\lambda^4t)}e_2.
$$
Hence, $a^{(1)}(\lambda) = 1$ and $b^{(1)}(\lambda) = 0$. Now we set
\begin{equation}
\label{back-2}
\eta^{(1)}(t,x) = \frac{1}{\gamma_1\overline{\lambda}_1} e^{-i(\lambda^2x+2\lambda^4t)} e_1
+\frac{1}{\overline{\lambda}_1} e^{i(\lambda^2x+2\lambda^4t)} e_2.
\end{equation}
By Lemma \ref{eta-by-new-jost}, the potential $u_{\lambda_1,\gamma_1}$, which we started with,
can be recovered by means of the inverse B\"{a}cklund transformation
\begin{equation}
\label{back-3}
u_{\lambda_1,\gamma_1} = \mathbf{B}_{\lambda_1}(\eta^{(1)})0.
\end{equation}
Explicit calculations with (\ref{back-2}) and (\ref{back-3}) yield the explicit expression
\begin{multline}\label{one-soliton}
    u_{\lambda_1,\gamma_1}(t,x) = 2i(\lambda_1^2-\overline{\lambda}_1^2)
\frac{\overline{\gamma}_1}{|\gamma_1|} \frac{e^{-2i(\lambda_1^2x+2\lambda_1^4t)}} {|e^{-i(\lambda_1^2x+2\lambda_1^4t)} |^2}\times\\
\frac{\overline{\lambda}_1|\gamma_1|^{-1}|e^{-i(\lambda_1^2x+2\lambda_1^4t)}|^2+ \lambda_1|\gamma_1||e^{i(\lambda_1^2x+2\lambda_1^4t)}|^2}{
(\lambda_1|\gamma_1|^{-1}| e^{-i(\lambda_1^2x+2\lambda_1^4t)}|^2+ \overline{\lambda}_1|\gamma_1| |e^{i(\lambda_1^2x+2\lambda_1^4t)}|^2)^2},
\end{multline}
which coincides with the one-soliton of the DNLS equation in the literature (see e.g. \cite{KN78}).

\begin{remark}
It is less straightforward to find the explicit expressions for the Jost functions of the linear system (\ref{laxeq1}) and (\ref{laxeq2})
with the one-soliton potential $u_{\lambda_1,\gamma_1}$ because the expressions (\ref{new-jost-1})--(\ref{new-jost-4})
can only be used in one way from $\{ \varphi_{\pm},\phi_{\pm}\}$ to $\{ \varphi_{\pm}^{(1)},\phi_{\pm}^{(1)}\}$,
which is hard to invert.
\end{remark}

\begin{remark}
For sake of completeness, we can rewrite the one-soliton solution (\ref{one-soliton}) in physical notations.
By defining
\begin{equation*}
        \omega=4|\lambda_1|^4,\quad
         v=-4{\rm Re}(\lambda_1^2)
\end{equation*}
and
\begin{equation*}
    x_0=\frac{2\ln(|\gamma_1|)}{\sqrt{4\omega-v^2}}, \quad \delta=\arg(\gamma_1)+\pi+3 \arctan\left(\frac{{\rm Im}(\lambda_1)} {{\rm Re}(\lambda_1)}\right)
\end{equation*}
with the obvious constraint $4 \omega - v^2 > 0$, $u_{\lambda_1,\gamma_1}$ is rewritten in the form used in \cite{CO06}:
\begin{equation}\label{one-soliton3}
    u_{\lambda_1,\gamma_1}(t,x)=\phi_{\omega,v}(x-vt-x_0) e^{-i\delta+i\omega t+i\frac{v}{2}(x-vt)- \frac{3}{4}i\int_{\infty}^{x-vt-x_0} |\phi_{\omega,v}(y)|^2dy},
\end{equation}
where
\begin{equation*}
    \phi_{\omega,v}(x)=\left[\frac{2 \sqrt{\omega} \cosh(\sqrt{4\omega-v^2}x)- v}{2(4\omega-v^2)} \right]^{-1/2}.
\end{equation*}
\end{remark}

By the computations in Lemma \ref{nonzero-a} and Remark \ref{new-b}, we obtain
\begin{equation}
\label{a-b-one-soliton}
a(\lambda) = \frac{\overline{\lambda}_1^2}{\lambda_1^2} \frac{\lambda^2-\lambda_1^2}{\lambda^2-\overline{\lambda}_1^2}, \quad
b(\lambda) = 0,
\end{equation}
for the one-soliton potential $u_{\lambda_1,\gamma_1}$. We also find
\begin{subequations}\label{L^2 norm}
\begin{eqnarray}
   \|u_{\lambda_1,\gamma_1} \|^2_{L^2} &=& 2 \sqrt{4 \omega - v^2} \int_{-\infty}^{\infty} \frac{dz}{2 \sqrt{\omega} \cosh(z) - v} \\
   &=&  8\arctan\left(\frac{2\sqrt{\omega}+v} {2\sqrt{\omega}-v}\right)
   \label{integral} \\
   &=& 8\arctan\left(\frac{{\rm Im}(\lambda_1)} {{\rm Re}(\lambda_1)}\right)\label{arg lambda 1}
   \\&=&8 \arg(\lambda_1),\label{arg lambda 2}
\end{eqnarray}
\end{subequations}
where we have used an explicit integral formula from \cite[Section 2.451]{Gradshteyn-Ryzhik} in order to obtain (\ref{integral}).
The equality between (\ref{arg lambda 1}) and (\ref{arg lambda 2}) holds because of $\lambda_1\in \C_{I}$.
The computation (\ref{L^2 norm}) confirms the asymptotic limit in Proposition \ref{prop-scattering}:
$$
a_{\infty} = \lim_{|\lambda|\rightarrow \infty}a(\lambda) = \frac{\overline{\lambda}_1^2}{\lambda_1^2} =
e^{-4 i \arg(\lambda_1)} = e^{\frac{1}{2i}\|u_{\lambda_1,\gamma_1} \|^2_{L^2}}.
$$

By using the representation (\ref{new-eig1}) in Lemma \ref{pull-back-eta}, the explicit formula (\ref{back-2}),
as well as the relation $d_{\lambda_1}(\eta,\eta) =[d_{\lambda_1}(\eta^{(1)},\eta^{(1)})]^{-1}$,
we can also find the function $\eta =(\eta_1,\eta_2)^t$ in the transformation (\ref{back-1}):
\begin{equation}
\label{eta-one-soliton-1}
\eta_1(t,x) = \frac{\overline{\lambda}_1 e^{-i(\overline{\lambda}_1^2 x + 2\overline{\lambda}_1^4 t )}}{
\lambda_1|e^{-i(\lambda_1^2x+2\lambda_1^4t)}|^2+\overline{\lambda}_1|e^{i(\lambda_1^2x+2\lambda_1^4t)}|^2}
\end{equation}
and
\begin{equation}
\label{eta-one-soliton-2}
\eta_2(t,x) = \frac{\overline{ \lambda}_1 e^{i(\overline{\lambda}_1^2x + 2 \overline{\lambda}_1^4 t ) } }{
\overline{\lambda}_1|e^{-i(\lambda_1^2x+2\lambda_1^4t)}|^2+\lambda_1|e^{i(\lambda_1^2x+2\lambda_1^4t)}|^2},
\end{equation}
where $\gamma_1 = 1$ is set for convenience. Since
\begin{equation*}
  d_{\lambda_1}(\eta,\eta)=\frac{|\lambda_1|^2}{\lambda_1|e^{-i(\lambda_1^2x+2\lambda_1^4t)}|^2+\overline{\lambda}_1|e^{i(\lambda_1^2x+2\lambda_1^4t)}|^2 }
\end{equation*}
satisfies the constraint
$$
-d_{\overline{\lambda}_1}(\eta,\eta)u_{\lambda_1,\gamma_1}+ 2i(\lambda_1^2-\overline{\lambda}_1^2) \eta_1\overline{\eta}_2=0,
$$
we confirm the transformation (\ref{back-1}) by using (\ref{one-soliton}), (\ref{eta-one-soliton-1}), and (\ref{eta-one-soliton-2}).

\section{Proof of Theorem \ref{theorem-main}}
Let $u_0 \in Z_1 \subset H^2(\mathbb{R}) \cap H^{1,1}(\mathbb{R})$
and $\lambda_1 \in \C_I$ be the only root of $a(\lambda)$ in $\C_I$.
By Lemma \ref{nonzero-a}, if $\eta(x) = \varphi_-(x;\lambda_1) e^{-i \lambda_1^2 x}$,
where $\varphi_-$ is the Jost function of the KN spectral problem (\ref{laxeq1})
associated with $u_0$, then $u_0^{(1)} = \mathbf{B}_{\lambda_1}(\eta)u_0$ belongs to $Z_0 \subset H^2(\mathbb{R}) \cap H^{1,1}(\mathbb{R})$.
Also, by Lemmas \ref{BT-lemma1}, \ref{pull-back-eta}, and Lemma \ref{eta-by-new-jost}, the mapping
is invertible with $u_0 = \mathbf{B}_{\lambda_1}(\eta^{(1)})u_0^{(1)}$, where $\eta^{(1)}$ is
expressed from the new Jost functions $\varphi^{(1)}_-$ and $\phi^{(1)}_+$
by the decomposition formula \eqref{eta-decomp}.

Let $T>0$ be a maximal existence time for the solution $u(t,\cdot) \in Z_1$, $t \in (-T,T)$ to the
Cauchy problem (\ref{dnls}) with the initial data $u_0 \in Z_1$ and eigenvalue $\lambda_1$.
For every fixed $t\in (-T,T)$, the solution $u(t,\cdot) \in Z_1$
admits the Jost functions $\{\varphi_{\pm}(t,x;\lambda),\phi_{\pm}(t,x;\lambda)\}$.
For every $t\in (-T,T)$, define $u^{(1)}$ by the B\"{a}cklund transformation
$$
u^{(1)} := \mathbf{B}_{\lambda_1}(\eta) u, \quad \eta(t,x) := \varphi_-(t,x;\lambda_1) e^{-i(\lambda_1^2x+2\lambda_1^4 t)},
$$
where we have used the boundary conditions (\ref{jost-infinity-original})
in the definition of $\varphi_-(t,x;\lambda_1)$ for every $t \in (-T,T)$.

By construction (see Corollary \ref{corollary}), $u^{(1)}(t,\cdot) \in Z_0$, $t \in (-T,T)$ is a solution of the Cauchy problem (\ref{dnls})
with the initial data $u^{(1)}_0 \in Z_0$. By existence and uniqueness theory \cite{Perry,Pelinovsky-Shimabukuro},
the solution $u^{(1)}(t,\cdot) \in Z_0$ is uniquely continued for every $t \in \mathbb{R}$.
Let $\{\varphi_{\pm}^{(1)}(t,x;\lambda),\phi_{\pm}^{(1)}(t,x;\lambda)\}$
be the Jost functions for $u^{(1)}(t,x)$. For every $t\in (-T,T)$, we have
$u = \mathbf{B}_{\lambda_1}(\eta^{(1)})u^{(1)}$ with
$$
\eta^{(1)}(t,x) = \frac{e^{-i(\lambda_1^2x+2\lambda_1^4 t)}}{\gamma_1\overline{\lambda}_1 a^{(1)}(\lambda_1)} \varphi_-^{(1)}(t,x;\lambda_1)+
\frac{e^{i(\lambda_1^2x+2\lambda_1^4 t)} }{\overline{\lambda}_1 a^{(1)}(\lambda_1)} \phi_+^{(1)}(t,x;\lambda_1),
$$
where $a^{(1)}(\lambda_1) \neq 0$ thanks to Lemma \ref{nonzero-a}.

On the other hand, since $u^{(1)}(t,\cdot)\in Z_0$ exists for every $t \in \mathbb{R}$, the associated Jost functions
$\{\varphi_{\pm}^{(1)}(t,x;\lambda), \phi_{\pm}^{(1)}(t,x;\lambda)\}$ exist for every
$t\in \mathbb{R}$ so that we can define
$$
\tilde{u} = \mathbf{B}_{\lambda_1}(\eta^{(1)})u^{(1)} \quad t \in \mathbb{R}.
$$
Since $u(t,\cdot)=\tilde{u}(t,\cdot) \in Z_1$ for every $t \in (-T,T)$ by uniqueness, the extended function $\tilde{u}$
is an unique extension of the solution $u$ to the same Cauchy problem (\ref{dnls})
that exists globally in time thanks to the bound (\ref{estimate-key}) proven in Lemma \ref{inverse BT-lemma1}.
Indeed, by \cite{Perry,Pelinovsky-Shimabukuro} we have $\|u^{(1)}(t,\cdot)\|_{H^2\cap H^{1,1}}\leq M_T$
for every $t\in (-T,T)$, where $T > 0$ is arbitrary and $M_T$ depends on $T$.
Next, by bound (\ref{estimate-key}) we have $\|u(t,\cdot)\|_{H^2\cap H^{1,1}}\leq C_{M_T}$ for every $t\in (-T,T)$.
Thus, the solution can not blow up in a finite time and hence
there exists a unique global solution $u(t,\cdot) \in Z_1$, $t \in \mathbb{R}$
to the Cauchy problem (\ref{dnls}) for every $u_0 \in Z_1 \subset
H^2(\mathbb{R}) \cap H^{1,1}(\mathbb{R})$.

By iterating the B\"{a}cklund transformation $k$ times and by the same argument as above,
we obtain the global existence of $u(t,\cdot)\in Z_k \subset
H^2(\mathbb{R}) \cap H^{1,1}(\mathbb{R})$, $t \in \mathbb{R}$
from the global existence of $u^{(k)}(t,\cdot)\in Z_0 \subset
H^2(\mathbb{R}) \cap H^{1,1}(\mathbb{R})$, $t \in \mathbb{R}$.
This completes the proof of Theorem \ref{theorem-main}.

\appendix

\section{Useful properties of $d_{\lambda}$, $S_{\lambda}$, and $G_{\lambda}$}
\label{algebra}

Recall the definition (\ref{def-d}) for the bilinear form $d_{\lambda}$ acting on $\C^2$ for a fixed $\lambda \in \C$.
One can easily verify the useful algebraic properties of $d_{\lambda}$ for every $\eta \in \C^2$
and $a,b \in \C$:
\begin{eqnarray}
& \phantom{t} &
d_{\lambda}(e_1,e_1) = \lambda, \quad d_{\lambda}(e_2,e_2) = \overline{\lambda}, \label{property2}\\
& \phantom{t} &
\overline{d_{\lambda}(\eta,\eta)} = d_{\overline{\lambda}}(\eta,\eta), \quad
d_{\lambda}(a \eta,b \eta) = a \overline{b} d_{\lambda}(\eta,\eta), \label{property4}\\
& \phantom{t} &
d_{\lambda}(\sigma_3\eta,\sigma_3\eta) = d_{\lambda}(\eta,\eta), \quad
d_{\lambda}(\sigma_1\eta,\sigma_1\eta) = d_{\overline{\lambda}}(\eta,\eta). \label{property6}
\end{eqnarray}
where $e_1=(1,0)^t$ and $e_2=(0,1)^t$ are basis vectors in $\C^2$, whereas
$\sigma_1$ and $\sigma_3$ are Pauli matrices given by (\ref{Pauli}).

Next, we recall the definition (\ref{def-G-S}) of the operators $G_{\lambda}$ and $S_{\lambda}$
acting on $\C^2$ for a fixed $\lambda \in \C$. From \eqref{property2} and \eqref{property4},
$G_{\lambda}$ and $S_{\lambda}$ satisfy for every $\eta \in \C^2$ and nonzero $a \in \C$:
\begin{equation} \label{Gproperty1}
G_{\lambda}( e_1)=\frac{\overline{\lambda}}{\lambda}, \quad G_{\lambda}( e_2)=\frac{\lambda}{\overline{\lambda}},\quad
\overline{G_{\lambda}(\eta)}=G_{\overline{\lambda}}(\eta), \quad G_{\lambda}(a\eta)=G_{\lambda}(\eta)
\end{equation}
and
\begin{equation} \label{Sproperty1}
S_{\lambda}(e_1)=S(e_2)=0,\quad S_{\lambda}(a \eta)=S_{\lambda}(\eta).
\end{equation}
From \eqref{property6}, we also have
\begin{equation}\label{Gproperty2}
G_{\lambda}(\sigma_3\eta)=G_{\lambda}(\eta),\quad G_{\lambda}(\sigma_1\eta)=G_{\overline{\lambda}}(\eta)
\end{equation}
and
\begin{equation}\label{Sproperty2}
-S_{\lambda}(\sigma_3\eta)=S_{\lambda}(\eta),\quad S_{\lambda}(\sigma_1\eta)=\overline{S_{\lambda}(\eta)}.
\end{equation}

By using \eqref{Gproperty1}, one can verify the following properties for every $\lambda \neq \pm \lambda_1$:
\begin{equation}
\label{G-comp}
\frac{\lambda^2\frac{\lambda_1}{\overline{\lambda}_1}G_{\lambda_1}(e_1)-\lambda_1^2}{\lambda^2-\lambda_1^2}=1, \quad
\frac{\lambda^2\frac{\lambda_1}{\overline{\lambda}_1}G_{\lambda_1}(e_2)-\lambda_1^2}{\lambda^2-\lambda_1^2}=
\frac{\lambda_1^2}{\overline{\lambda}_1^2}\frac{\lambda^2-\overline{\lambda}_1^2}{\lambda^2-\lambda_1^2}
\end{equation}
and
\begin{equation}
\label{G-comp-again}
\frac{\lambda^2\frac{\lambda_1}{\overline{\lambda}_1}G_{\overline{\lambda}_1}(e_2)-\lambda_1^2}{\lambda^2-\lambda_1^2}=1, \quad
\frac{\lambda^2\frac{\lambda_1}{\overline{\lambda}_1}G_{\overline{\lambda}_1}(e_1)-\lambda_1^2}{\lambda^2-\lambda_1^2}=
\frac{\lambda_1^2}{\overline{\lambda}_1^2}\frac{\lambda^2-\overline{\lambda}_1^2}{\lambda^2-\lambda_1^2}.
\end{equation}

\section{On regularity of Jost functions}
\label{appendix-Jost}

Recall that if $u \in H^{1,1}(\mathbb{R})$, then $u \in L^1(\mathbb{R}) \cap L^{\infty}(\mathbb{R})$ and
$\partial_x u \in L^1(\mathbb{R})$, so that the assumptions of Propositions \ref{KN-existence} and \ref{prop-Jost-KN}
are satisfied. In what follows, we establish more regularity results for Jost functions compared
to what was established previously in \cite{Pelinovsky-Shimabukuro}.
\\ \phantom{a}\\
{\bf Lemma B. }{\em
For every $u\in H^2(\mathbb{R}) \cap H^{1,1}(\mathbb{R})$ satisfying
$\|u\|_{H^{2}(\mathbb{R})\cap H^{1,1}(\mathbb{R})}\leq M$ for some $M > 0$,
let $\varphi_{\pm}(x;\lambda) e^{-i\lambda^2x}$ and
$\phi_{\pm}(x;\lambda) e^{+i\lambda^2x}$ be Jost functions of the KN spectral problem \eqref{laxeq1} given in
Propositions \ref{KN-existence} and \ref{prop-Jost-KN}. Fix $\lambda_1\in \C$ satisfying $\mbox{Im}(\lambda_1^2)>0$ and denote $\varphi_- := \varphi_-(\cdot;\lambda_1) = (\varphi_{-,1},\varphi_{-,2})^t$ and
$\phi_+ := \phi_+(\cdot;\lambda_1) = (\phi_{+,1},\phi_{+,2})^t$. Then,
\begin{equation} \label{norm1}
\|\langle x \rangle \varphi_{-,2}\|_{L^2(\R)}+
\|\langle x \rangle \partial_x\varphi_-\|_{L^2(\R)}+
\|\langle x \rangle \partial_x^2\varphi_-\|_{L^2(\R)}+
\| \partial_x^3\varphi_-\|_{L^2(\R)}\leq C_M,
\end{equation}
and
\begin{equation} \label{norm2}
\|\langle x \rangle \phi_{+,1}\|_{L^2(\R)}+
\|\langle x \rangle \partial_x\phi_+\|_{L^2(\R)}+
\|\langle x \rangle \partial_x^2\phi_+\|_{L^2(\R)}+
\| \partial_x^3\phi_+\|_{L^2(\R)}\leq C_M,
\end{equation}
where the constant $C_M$ does not depend on $u$.}

\begin{proof}
We will prove the statement for $\varphi_-$ since the proof for $\phi_+$ is similar.
From Proposition \ref{KN-existence}, we know that $\varphi_-$ belongs to $L^{\infty}(\mathbb{R})$.
Let us first show that the second component $\varphi_{-,2}$ is square integrable.
Compared with Lemma 1 in \cite{Pelinovsky-Shimabukuro},
where the existence of Jost functions is proved uniformly in $\lambda$, the assertion of this
proposition is easier to prove for just one $\lambda = \lambda_1$.
We can work with the integral equation for $\varphi_-$:
$$
\varphi_-=e_1+K\varphi_-,
$$
where the operator $K$ is given as
$$
K\varphi_-= \lambda_1 \int_{-\infty}^x  \left[\begin{matrix} 1 & 0 \\ 0 & e^{2i \lambda_1^2 (x-y)}
\end{matrix}\right] Q(u(y)) \varphi_{-}(y;\lambda) dy.
$$
This integral operator can be bounded as follows,
$$
\left[\begin{matrix} \|(K \varphi_-)_1\|_{L^{\infty}(-\infty,x_0)} \\ \|(K \varphi_-)_2\|_{L^{2}(-\infty,x_0)} \end{matrix}\right] \leq |\lambda_1|\|u\|_{L^2(-\infty,x_0)} \left[\begin{matrix} 0 & 1 \\ \frac{1}{2\text{Im}(\lambda^2_1)} & 0\end{matrix}\right] \left[\begin{matrix} \|\varphi_{-,1}\|_{L^{\infty}(-\infty,x_0)} \\ \|\varphi_{-,2}\|_{L^{2}(-\infty,x_0)} \end{matrix}\right].
$$
Thus, we deduce that for every fixed $\lambda_1$ satisfying $\mbox{Im}(\lambda_1^2)>0$,
there exists $x_0\in \mathbb{R}$ such that $K$ is a contraction on $L^{\infty}(-\infty,x_0)\times L^2(-\infty,x_0).$
Since $u\in L^2(\mathbb{R})$, we can divide $\mathbb{R}$ into finitely many subintervals such that
$K$ is a contraction as shown above within each subinterval. By patching solutions together,
we obtain that $\varphi_{-,2}$ belongs to $L^2(\mathbb{R})$ and satisfies
\begin{equation}\label{L^2 bound for varphi}
  \|\varphi_{-,2}\|_{L^2(\R)} \leq C_M \|u\|_{L^2(\R)}
\end{equation}
where $C_M$ does not depend on $u$.

Next, it follows directly from the Kaup--Newell system \eqref{laxeq1} that
$$
\partial_x\varphi_{-,1} = \lambda_1 u \varphi_{-,2} \Longrightarrow \partial_x\varphi_{-,1} \in L^2(\mathbb{R})
$$
and
$$
\partial_x\varphi_{-,2} = -\lambda_1 \overline{u}\varphi_{-,1} + 2i \lambda_1^2\varphi_{-,2} \Longrightarrow
\partial_x\varphi_{-,2} \in L^2(\mathbb{R}).
$$
Differentiating (\ref{laxeq1}) once and twice, we also obtain
$\partial_x^2 \varphi_-, \partial_x^3 \varphi_- \in L^2(\mathbb{R})$.

In order to show $x\varphi_{-,2} \in L^2(\mathbb{R})$, we write
$$
\partial_x(x\varphi_{-,2})= \varphi_{-,2} + x\partial_x\varphi_{-,2}
$$
and use the second component of the Kaup--Newell system \eqref{laxeq1} to get
$$
\partial_x(x\varphi_{-,2}) = 2i\lambda_1^2 x \varphi_{-,2} + \varphi_{-,2} - \lambda x\overline{u}\varphi_{-,2}
$$
which yields the integral equation
$$
x\varphi_{-,2}(x) = \int_{-\infty}^x e^{2i\lambda_1^2(x-y)} \varphi_{-,2}(y)dy  -
\lambda_1 \int_{-\infty}^x e^{2i\lambda_1^2(x-y)} y \overline{u}(y) \varphi_{-,2}(y) dy.
$$
Since the right hand side is bounded in $L^{2}(\mathbb{R})$, we have $x\varphi_{-,2}$
in $L^2(\mathbb{R})$. Then, it follows from system (\ref{laxeq1}) and its derivative
that $x \partial_x \varphi_-, x \partial_x^2 \varphi_- \in L^{2}(\mathbb{R})$.
Combining all estimates together, we obtain bounds (\ref{norm1}) for $\varphi_-$.
\end{proof}


\begin{thebibliography}{10}
\bibitem{Beals1984} R. Beals and R. R. Coifman, ``Scattering and inverse scattering for first order systems",
Comm. Pure Appl. Math. {\bf 37} (1984), 39--90.

\bibitem{CO06} M. Colin and M. Ohta, ``Stability of solitary waves for derivative nonlinear Schr\"{o}dinger equation",
Ann. I.H. Poincar\'e-AN {\bf 23} (2006), 753--764.

\bibitem{Sulem3} Y. Cher, G. Simpson, and C. Sulem, ``Local structure of singular profiles for a derivative
nonlinear Schr\"{o}dinger equation", SIAM J. Appl. Dyn. Syst. {\bf 16} (2017), 514--545.

\bibitem{C-P} S. Cuccagna and D.E. Pelinovsky, ``The asymptotic stability of solitons
in the cubic NLS equation on the line", Applicable Analysis, {\bf 93} (2014), 791--822.

\bibitem{D-J} P.A. Deift and J. Park, ``Long-time asymptotics for solutions of the NLS equation
with a delta potential and even initial data", Int. Math. Res. Notes {\bf 24} (2011), 5505--5624.

\bibitem{D-Z-1} P.A. Deift and X. Zhou, ``Long-time asymptotics for solutions
of the NLS equation with initial data in weighted Sobolev spaces", Comm. Pure Appl. Math.
{\bf 56} (2003), 1029--1077.

\bibitem{FHT} N. Fukaya, M. Hayashi, and T. Inui, ``A sufficient condition for global existence
of solutions to a generalized derivative nonlinear Schr\"{o}dinger equation",
Analysis \& PDE {\bf 10} (2017), 1149--1167.

\bibitem{Gradshteyn-Ryzhik} I.S. Gradshteyn and I.M. Ryzhik,
``Tables of Integrals, Series, and Products (Seventh Edition)", Academic Press (2007).

\bibitem{Guo-Ling-Liu} B. Guo, L. Ling, and Q. P. Liu, ``High-order solutions and generalized Darboux transformations of
derivative nonlinear Schr\"{o}dinger equations", Stud. Appl. Math. {\bf 130} (2013), 317-344.

\bibitem{GuoWu1995} B. Guo and Y. Wu, ``Orbital stability of solitary waves for the nonlinear derivative
Schr\"{o}dinger equation", J. Diff. Eqs. {\bf 123} (1995), 35--55.

\bibitem{Hayashi} N. Hayashi, ``The initial value problem for the derivative nonlinear Schr\"{o}dinger equation in the energy space",
Nonlinear Anal. {\bf 20} (1993), 823--833

\bibitem{H-O-1} N. Hayashi and T. Ozawa, ``On the derivative nonlinear Schr\"{o}dinger equation", Physica D {\bf 55} (1992), 14--36.

\bibitem{Imai} K. Imai, ``Generalization of Kaup-Newell inverse scattering formulation and Darboux transformation'',
J. Phys. Soc. Japan. {\bf 68} (1999), 355-359.

\bibitem{KN78} D. Kaup and A. Newell, ``An exact solution for a derivative nonlinear Schr\"{o}dinger equation",
J. Math. Phys. {\bf 19} (1978), 789--801.

\bibitem{WuKwon} S. Kwon and Y. Wu, ``Orbital stability of solitary waves for derivative nonlinear Schr\"{o}dinger equation",
arXiv:1603.03745 (2016)

\bibitem{LeCoz} S. Le Coz and Y. Wu, ``Stability of multisolitons for the derivative nonlinear Schr\"{o}dinger equation",
Inter Math Res Notices (2017), doi: 10.1093/imrn/rnx013.

\bibitem{Perry} J. Liu, P.A. Perry, and C. Sulem, ``Global existence for the derivative nonlinear Schr\"{o}dinger
equation by the method of inverse scattering", Communications in PDEs {\bf 41} (2016), 1692--1760.

\bibitem{Sulem1} X. Liu, G. Simpson, and C. Sulem, ``Stability of solitary waves for a generalized derivative nonlinear Schr\"{o}dinger
equation", J. Nonlin. Science {\bf 23} (2013), 557--583.

\bibitem{Sulem2} X. Liu, G. Simpson, and C. Sulem, ``Focusing singularity in a derivative nonlinear Schr\"{o}dinger
equation", Physica D {\bf 262} (2013), 48--58.

\bibitem{MWX1} C. Miao, X. Tang, and G. Xu,  ``Solitary waves for nonlinear Schr\"{o}dinger equation with derivative",
Commun. Contemp. Math. (2017), doi: 10.1142/S0219199717500493.

\bibitem{MWX2} C. Miao, X. Tang, and G. Xu,  ``Stability of the traveling waves for the derivative Schr\"{o}dinger equation
in the energy space", Calc. Var. PDEs {\bf 56} (2017) Article:45.

\bibitem{Pelinovsky-Shimabukuro} D. E. Pelinovsky and Y. Shimabukuro, ``Existence of global solutions
to the derivative NLS equation with the inverse scattering transform method",
Inter Math Res Notices (2017), doi: 10.1093/imrn/rnx051.

\bibitem{Steudel} H. Steudel, ``The hierarchy of multi-soliton solutions of the derivative
nonlinear Schr\"{o}dinger equation", J. Phys. A {\bf 36} (2003), 1931--1946.

\bibitem{MWX3} X. Tang and G. Xu,  ``Stability of the sum of two solitary waves for (GDNLS) in the energy space", arXiv: 1702.07858 (2017).

\bibitem{TF1} M. Tsutsumi and I. Fukuda, ``On solutions of the
derivative nonlinear Schr\"{o}dinger equation. Existence and uniqueness theorem", Funkcialaj Ekvacioj
{\bf 23} (1980), 259--277.

\bibitem{TF2} M. Tsutsumi and I. Fukuda, ``On solutions of the
derivative nonlinear Schr\"{o}dinger equation II", Funkcialaj Ekvacioj
{\bf 24} (1981), 85--94.

\bibitem{W} Y. Wu, ``Global well-posedness for the nonlinear Schr\"{o}dinger equation with
derivative in energy space", Analysis \& PDE {\bf 6} (2013), 1989--2002.

\bibitem{W-2} Y. Wu, ``Global well-posedness on the derivative nonlinear Schr\"{o}dinger equation",
Analysis \& PDE  {\bf 8} (2015), 1101--1112.

\bibitem{Xu-He-Wang} S. Xu, J. He, and L. Wang, ``The Darboux transformation of the derivative
nonlinear Schr\"{o}dinger equation", J. Phys. A {\bf 44} (2011), 305203

\bibitem{Z} X. Zhou, ``$L^2$-Sobolev space bijectivity of the scattering and inverse scattering transforms",
Comm. Pure Appl. Math., {\bf 51} (1998), 697--731.

\end{thebibliography}
\end{document}